\documentclass{lmcs} 
\pdfoutput=1
\usepackage[utf8]{inputenc}

\usepackage{lastpage}
\lmcsdoi{21}{2}{4}
\lmcsheading{}{\pageref{LastPage}}{}{}%
{Dec.~22,~2023}{Apr.~16,~2025}{}

\keywords{graph transformation, moving, restriction, derivation, adhesive categories}

\usepackage{hyperref}

\usepackage{graphicx}        
\usepackage{subcaption}
\makeatletter         
\def\figurecaption#1#2{\noindent\hangindent 40pt
                       \hbox to 36pt {\small\sl #1 \hfil}
                       \ignorespaces {\small #2}}

\long\def\@makecaption#1#2{
  \vskip 10pt 
  \settowidth{\@tempdima}{#2}
  \ifdim\@tempdima>0pt
       \setbox\@tempboxa\hbox{#1: #2}
     \else
       \setbox\@tempboxa\hbox{#1 #2}
   \fi
   \ifdim \wd\@tempboxa >\hsize               
       \begin{list}{#1:}{
       \settowidth{\labelwidth}{#1:}
       \setlength{\leftmargin}{\labelwidth}
       \addtolength{\leftmargin}{\labelsep}
        }\item #2 \end{list}\par   
     \else                                    
       \hbox to\hsize{\hfil\box\@tempboxa\hfil}  
   \fi}
\makeatother

\usepackage{amsfonts}
\usepackage{multirow}
\usepackage{tikz}
\usetikzlibrary{patterns,%
								positioning,%
								calc,%
								scopes,%
								decorations.fractals,%
								shapes.geometric,%
								shadows,%
								decorations.pathreplacing,%
								decorations.pathmorphing,%
								decorations.markings,%
								decorations,%
								shapes,%
								arrows,
       automata,
       petri,
								snakes,
       lindenmayersystems,
       cd
}
\tikzstyle{every picture} = [every state/.style = {circle,draw,minimum size = 0.15cm, inner sep = 1pt}, every edge/.style = {draw,inner sep = 2pt},on grid,inner sep = 0.03cm, every loop/.style={}, node distance = 0.4cm]
\tikzset{
    wavy/.style = {->,very thin,double distance = 0,shorten <=.2pt,>=stealth',snake=snake, segment amplitude=1pt, segment length=8pt, line after snake=4pt}
}
\usepgflibrary{shapes.geometric}

\usepackage{xifthen}
\usepackage{xspace}

\usepackage{latexsym,amssymb,amsmath}
\usepackage{xcolor}

\usepackage{hyperref}
\usepackage{tabularx}
\usepackage{enumitem}

\newcommand{\cat}[1]{\mathbf{#1}}

\newcommand{\dder}{\mathop{\Longrightarrow}\limits}

\newcommand{\N}{\mathbb N}

\newcommand{\iso}{\cong}
\newcommand{\id}{\mathit{id}}

\newcommand{\G}{\mathcal{G}}

\newcommand{\base}{\noindent Base:\xspace }
\newcommand{\step}{\noindent Step:\xspace }

\newcommand{\conflux}{\mathit{conflux}}
\newcommand{\interchange}{\mathit{interchange}}

\newcommand{\move}{\mathit{move}}
\newcommand{\evom}{\mathit{evom}}
\newcommand{\acc}{\mathit{acc}}
\newcommand{\ACC}{\mathit{ACC}}
\newcommand{\spine}{\mathit{spine}}
\newcommand{\head}{\mathit{head}}
\newcommand{\tail}{\mathit{tail}}

\newcommand{\colr}{\mathit{color}}
\newcommand{\dual}{\mathit{dual}}
\newcommand{\addloop}{\mathit{add\_loop}}
\newcommand{\addcolor}{\mathit{add\_color}}
\newcommand{\choosecolr}{\mathit{choose\_color}}
\newcommand{\doubleedge}{\mathit{double\_edge}}
\newcommand{\addedge}{\mathit{add\_edge}}
\newcommand{\removepair}{\mathit{remove\_pair}}

\theoremstyle{plain}\newtheorem{construction}[defi]{Construction} 

\begin{document}

\title[Moving a Derivation Along a Derivation Preserves the Spine]{Moving a Derivation Along a Derivation\texorpdfstring{\\}{} Preserves the Spine in Adhesive Categories}

\author[H.-J.~Kreowski]{Hans-J\"org Kreowski\lmcsorcid{0000-0003-0578-5882}}[a]
\author[A.~Lye]{Aaron Lye\lmcsorcid{0000-0003-2987-8661}}[b]
\author[A.~Windhorst]{Aljoscha Windhorst\lmcsorcid{009-0007-2082-9552}}[a,b]

\address{University of Bremen, Department of Computer Science, P.O.Box 33 04 40, 28334 Bremen, Germany}	
\email{$\{$kreo,windhorst$\}$@uni-bremen.de}  

\address{Institute for the Protection of Maritime Infrastructures, German Aerospace Center (DLR), Fischkai 1, 27572, Bremerhaven, Germany}	
\email{$\{$aaron.lye,aljoscha.windhorst$\}$@dlr.de}  





\begin{abstract}
  \noindent
  In this paper, we investigate the relationship between two elementary operations on
derivations in the framework of graph transformation based on adhesive categories:
moving a derivation along a derivation based on parallel and sequential independence on one hand and restriction of a derivation with respect to a monomorphism into the start object on the other hand.
Intuitively, a restriction clips off parts of the start object that are never matched by a rule application throughout the derivation on the other hand.
As main result, it is shown that moving a derivation preserves its spine being the minimal restriction.

\end{abstract}

\maketitle

\section{Introduction}
\label{sec:introduction}
A major part of the study of graph transformation concerns the fundamental properties of derivations.
In this paper, we contribute to this study by relating two operations on derivations both based on long-known concepts: \emph{moving} and \emph{restriction}.
Moreover, the investigation is conducted using adhesive categories.

The first operation, introduced in Section~\ref{sec:moving-derivations}, moves derivations along derivations.
We consider two variants: \emph{forward moving} and \emph{backward moving}.
It is well-known and often considered that two applications of rules to the same graph
or to an object in an adhesive category respectively
can be applied one after the other with the same result if they are parallel independent. 
As an operation, this strong confluence may be called \emph{conflux}. 
Analogously, two successive rule applications can be interchanged if they are sequentially independent, called \emph{interchange} as an operation. 
The situation is depicted in Figure~\ref{fig:conflux-interchange}.

\begin{figure}[h!]
\[
\begin{tikzcd}
G \arrow[r, Rightarrow, "\overline{r}" below] \arrow[d, Rightarrow, "r" right]
& \overline{H} \\
H
& 
\end{tikzcd}
~~~~~~
\begin{tikzpicture}[baseline=0pt]
\node (left) {};
\node (right) [right=30pt of left] {};
\path[draw,wavy] (left) -- node[below = .1]{\scriptsize $\conflux$} (right);
\end{tikzpicture}
~~~~~~
\begin{tikzcd}
G \arrow[r, Rightarrow, "\overline{r}" below] \arrow[d, Rightarrow, "r" right]
& \overline{H} \arrow[d, Rightarrow, "r" right] \\
H \arrow[r, Rightarrow, "\overline{r}" below]
& X
\end{tikzcd}
~~~~~~
\begin{tikzpicture}[baseline=0pt]
\node (left) {};
\node (right) [right=30pt of left] {};
\path[draw,wavy] (right) -- node[below = .1]{\scriptsize $\interchange$} (left);
\end{tikzpicture}
~~~~~~
\begin{tikzcd}
G \arrow[r, Rightarrow, "\overline{r}" below]
& \overline{H} \arrow[d, Rightarrow, "r" right] \\
& X
\end{tikzcd}
\]
\caption{\emph{conflux} and \emph{interchange}}
\label{fig:conflux-interchange}
\end{figure}
In the case of graphs, the confluence property of parallel independent rule applications was announced by Rosen in~\cite{Rosen:75} and proved by Ehrig and Rosen in~\cite{Ehrig-Rosen:76}; the interchangeability of sequentially independent rule applications was introduced and studied by Ehrig and Kreowski in~\cite{Ehrig-Kreowski:76} (see, e.g., Ehrig~\cite{Ehrig:79}, Corradini et al.~\cite{Corradini-Ehrig-Loewe.ea:97} and Ehrig et al.~\cite{Ehrig-Ehrig-Prange.ea:06} for comprehensive surveys).
Both results are also known as local Church-Rosser theorems.
They are considered in the framework of graph transformation based on adhesive categories in Ehrig et al.~\cite{Ehrig-Ehrig-Prange.ea:06,Ehrig-Ermel-Golas-Hermann:15}.

It is not difficult to show that the operations can be generalized to derivations. 
If $P$ and $\overline{P}$ are sets of rules such that each two applications of a $P$-rule and a $\overline{P}$-rule to the same object are parallel independent, then derivations $d=(G \dder^*_P H)$ and $\overline{d} = (G \dder^*_{\overline{P}} \overline{H})$ induce derivations $d' = (\overline{H} \dder^*_P X)$ and $\overline{d'} = (H \dder^*_{\overline{P}} X)$ for some object $X$ by iterating $\conflux$ as long as possible.
Analogously, two successive derivations $\overline{d} = (G \dder^*_{\overline{P}} \overline{H})$ and $d' = (\overline{H} \dder^*_P X)$ induce derivations $d = (G \dder^*_P H)$ and $\overline{d'} = (H \dder^*_{\overline{P}} X)$ for some object $H$ provided that each successive application of a $\overline{P}$-rule and a $P$-rule are sequentially independent. 
The situation is depicted in Figure~\ref{fig:moving}.
We interpret these two operations on derivations as moving $d$ forward along $\overline{d}$ (and symmetrically $\overline{d}$ along $d$) on one hand and as moving $d'$ backward along~$\overline{d}$ on the other hand.

\begin{figure}[h!]
\[
\begin{tikzcd}
G \arrow[r, Rightarrow, "*" above, "\overline{P}" below] \arrow[d, Rightarrow, "*" left, "P" right]
& \overline{H} \\
H
& 
\end{tikzcd}
~~~~~~
\begin{tikzpicture}[baseline=0pt]
\node (left) {};
\node (right) [right=30pt of left] {};
\path[draw,wavy] (left) -- node[below = .1]{\scriptsize forward} node[below = .35]{\scriptsize moving} (right);
\end{tikzpicture}
~~~~~~
\begin{tikzcd}
G \arrow[r, Rightarrow, "*" above, "\overline{P}" below] \arrow[d, Rightarrow, "*" left, "P" right]
& \overline{H} \arrow[d, Rightarrow, "*" left, "P" right] \\
H \arrow[r, Rightarrow, "*" above, "\overline{P}" below]
& X
\end{tikzcd}
~~~~~~
\begin{tikzpicture}[baseline=0pt]
\node (left) {};
\node (right) [right=30pt of left] {};
\path[draw,wavy] (right) -- node[below = .1]{\scriptsize backward} node[below = .35]{\scriptsize moving} (left);
\end{tikzpicture}
~~~~~~
\begin{tikzcd}
G \arrow[r, Rightarrow, "*" above, "\overline{P}" below]
& \overline{H} \arrow[d, Rightarrow, "*" left, "P" right] \\
& X
\end{tikzcd}
\]
\caption{Forward and backward moving}
\label{fig:moving}
\end{figure}

The second operation is introduced in Section~\ref{sec:accessed-part-and-spine}.
It restricts a derivation to a monomorphism of its start object provided that the \emph{accessed part} of the derivation
factors through the given monomorphism.
In the category of graphs, the accessed part consists of all vertices and edges of the start graph that are accessed by some of the matching morphisms of the left-hand sides of applied rules within the derivation. 
A \emph{restriction} clips off parts of the complement of the accessed part, i.e., an invariant part of all objects throughout the derivation. 
As the accessed part is a monomorphism factoring through itself, the restriction to it, which we call \emph{spine}, is minimal.
In case of graphs it consists only of vertices and edges essential for the derivation.
If $d'$ is a restriction of a derivation $d$, then $d$ can be considered as an extension of $d'$, and $d'$ is embedded into a larger context.
Extension and restriction are investigated in some detail by Ehrig~\cite{Ehrig:77}, Kreowski~\cite{Kreowski:78,Kreowski:79}, and Plump~\cite{Plump:99a}.

Finally, we state and prove the main result of this paper in Section~\ref{sec:moving-preserves-the-spine}: 
Forward moving and backward moving preserve the spine of the moved derivation. 
This investigation is motivated by our recently proposed graph-transformational approach for proving the correctness of reductions between NP-problems in~\cite{Kreowski-Kuske-Lye-Windhorst:22}. 
Such proofs require certain derivations over a set of rules $P$ to be constructed from certain derivations over a set of rules $P'$ where, in addition, the initial graphs of both are connected by a derivation over a set of rules $\overline{P}$.
We have provided a toolbox of operations on derivations that allow such constructions where one essential operation is moving.
Moreover, in some sample correctness proofs the phenomenon of the main result could be observed and used favorably.
Our running example stems from this context.

In \cite{Kreowski-Kuske-Lye-Windhorst:23}, we have shown that moving preserves the spine and provided all
the prerequisites for the category of directed edge-labeled graphs. To
define the accessed part of a derivation and the restriction of a
derivation and to prove the needed results for restrictions and the main
theorem, we have made heavy use of set-theoretical constructs like
subgraphs, intersections, unions and the constructions of pushout
complements and derived graphs by employing the difference of sets.
As in this paper, the whole consideration is set into the more general
framework of adhesive categories, it is necessary to replace the
set-theoretic features by categorical concepts. More explicitly,
subgraphs are replaced by monomorphisms, intersections by pullbacks of
monomorphisms and unions by pushout of pullbacks of monomorphisms. As
there is no categorical counterpart to the set difference, there is no
direct replacement. But certain pushout and pullback properties allow to
circumvent the problem.

Our main result for the category of graphs has been very helpful to prove the
correctness of various reductions between NP-problems for graphs. Many of
these problems like the problems of finding colorings, cliques, particular
paths and cycles, vertex covers, spanning trees, etc. make also sense for
hypergraphs and other graph-like structures. Hence, we hope that the
preservation of spines by moving based on adhesive categories will turn out to
be as useful in these cases as in the case of graphs.

Restrictions of derivations over adhesive categories are only mentioned as
inverse to extensions in Ehrig et al.~\cite{Ehrig-Ehrig-Prange.ea:06,Ehrig-Ermel-Golas-Hermann:15}. In Ehrig et al.~\cite{Ehrig-Ehrig-Prange.ea:06}, only the
comparatively simple single-step restriction is constructed explicitly. The
early investigations by Kreowski~\cite{Kreowski:78,Kreowski:79} and Plump~\cite{Plump:99a} provide constructions
of restrictions for graphs and hypergraphs respectively and some properties
concerning the iteration of restrictions and their relation to extensions. For
the category of graphs, a more detailed investigation of restrictions can be
found in our conference paper~\cite{Kreowski-Kuske-Lye-Windhorst:23}. In the latter four references, the
constructions are not fully categorical as specific properties of graphs and
hypergraphs are used.

\section{Preliminaries}
\label{sec:preliminaries}
In this section, the categorical prerequisites are provided.
For the well-known categorical notions including
initial objects, pullbacks, and pushouts 
cf., e.g.,
Adamek et al.~\cite{Adamek-Herrlich-Strecker:09} and
Ehrig et al.~\cite{Ehrig-Ehrig-Prange.ea:06,Ehrig-Ermel-Golas-Hermann:15}.
The concepts and properties
of adhesive categories as far as used in this paper are recalled
explicitly. Moreover, we recall the basics of the category of graphs as
a typical example of an adhesive category because our running examples
are set into this category.

\subsection{Graphs}
There are various categories that can be derived from $\cat{Sets}$ in such a way that they inherit the concepts from $\cat{Sets}$.
The category $\cat{Graphs}$ of (directed edge-labeled) \emph{graphs} over some alphabet $\Sigma$ is a well-known example.
Its objects are graphs, its morphisms are graph morphisms defined as follows:
Let $\Sigma$ be a set of labels with $* \in \Sigma$.
A (directed edge-labeled) \emph{graph} over $\Sigma$ is a system $G = (V,E,s,t,l)$ where $V$ is a finite  set of \emph{vertices}, $E$ is a finite set of \emph{edges}, $s,t\colon E\to V$ and $l\colon E\to\Sigma$ are mappings assigning a \emph{source}, a \emph{target} and a \emph{label} to every edge $e\in E$.
An edge $e$ with $s(e) = t(e)$ is  called a \emph{loop}.
An edge with label $*$ is called an \emph{unlabeled edge}.
In drawings, the label $*$ is omitted. 
\emph{Undirected edges} are pairs of edges between the same vertices in opposite directions.
The empty graph is denoted by $\emptyset$.
The class of all directed edge-labeled graphs is denoted by $\G_\Sigma$.
For graphs $G,H \in \G_{\Sigma}$, a \emph{graph morphism} $g\colon G\to H$ is a pair of mappings $g_V\colon V_G\to V_H$ and $g_E\colon E_G\to E_H$ that are structure-preserving, i.e.,
$g_V(s_G(e)) = s_H(g_E(e))$,
$g_V(t_G(e)) = t_H(g_E(e))$,
and $ l_G(e) = l_H(g_E(e))$ for all $e \in E_G$.
If the mappings $g_V$ and $g_E$ are bijective, then $G$ and $H$ are \emph{isomorphic}. 
If they are inclusions, then $G$ is called a \emph{subgraph} of $H,$ denoted by $G \subseteq H$.
For a graph morphism $g \colon G\to H$, the image of $G$ in $H$ is called a \emph{match} of $G$ in $H$, i.e., the match of $G$ with respect to the morphism $g$ is the subgraph $g(G) \subseteq H$.

Many graph-like structures behave like the category of graphs.
Capturing the required properties for rule-based transformation leads to the definition of adhesive categories.
We capture the needed categorical notions and facts in the following subsection.

\subsection{Categorical Notions and Facts}
\label{subsec:categorical-notions-and-facts}
We assume that the underlying category $\cat{C}$ is adhesive.
In addition, we require an epi-mono factorization and a strict initial object.
We are using various properties of such a category which are listed as facts.

\begin{asm}\hfill
\begin{enumerate}
\item
$\cat{C}$ is adhesive, i.e.,
\begin{enumerate}
\item
  $\cat{C}$ has all pullbacks,
\item
  $\cat{C}$ has pushouts along monomorphisms (meaning at least one of the two morphisms in the pushout span is a monomorphism), and
\item
pushouts along monomorphisms are Van Kampen squares, where 
a Van Kampen squares is a pushout of $A \xleftarrow[]{} C \xrightarrow[]{} B$ satisfying the following condition:
For any commutative cube for which the pushout square of $A \xleftarrow[]{} C \xrightarrow[]{} B$ forms the bottom face and the back faces are pullbacks, the front faces are pullbacks if and only if the top face is a pushout.
\[
\begin{tikzcd}[row sep=tiny,column sep=tiny]
& C'\arrow[dl]\arrow[rr]\arrow[dd] & & B'\arrow[dl]\arrow[dd] \\
A'\arrow[rr,crossing over]\arrow[dd] & & D' \\
& C\arrow[dl]\arrow[rr] &  & B\arrow[dl] \\
A\arrow[rr] & & D\arrow[from=uu,crossing over]
\end{tikzcd}
\]
  \end{enumerate}

\item
$\cat{C}$ has a strict initial object $\emptyset$, i.e., an initial object $\emptyset$ with the property that every morphism in $\cat{C}$ with codomain $\emptyset$ is an isomorphism.

\item
Every morphism in $\cat{C}$ has an epi-mono factorization, i.e., for every morphism $f$ there is a factorization $f = m \circ e$ where $e$ is an epimorphism and $m$ is a monomorphism.

\end{enumerate}
\end{asm}

Pushouts and pullbacks have useful properties.
\begin{fact}\hfill
\begin{enumerate}
\item
Monomorphisms are pullback stable, i.e.,
wrt the following diagram:
\[
\begin{tikzcd}
& A \arrow[r, "a"] \arrow[d, ""]
& B \arrow[d, ""]  \\
& C \arrow[r, "d"] 
& D 
\end{tikzcd}
\]
if the diagram is a pullback and $d$ is a monomorphism, then $a$ is a monomorphism.

\item
Isomorphisms are pullback and pushout stable.
\end{enumerate}
\end{fact}

Moreover, pushouts and pullbacks have nice composition and decomposition properties.

Consider a commuting diagram of the following shape:
\[
\begin{tikzcd}
A \arrow[r] \arrow[d]
& B \arrow[r] \arrow[d]
& C \arrow[d] \\
D \arrow[r]
& E \arrow[r]
& F
\end{tikzcd}
\]
\begin{fact}
\label{fact:composition-lemma-pb}
If the right square is a pullback, then the outer rectangle is a pullback if and only if the left square is a pullback (cf.~\cite[Fact 2.27]{Ehrig-Ehrig-Prange.ea:06}).
\end{fact}
\begin{fact}
\label{fact:composition-lemma-po}
If the left square is a pushout, then the outer rectangle is a pushout if and only if the right square is a pushout (cf.~\cite[Fact 2.20]{Ehrig-Ehrig-Prange.ea:06}).
\end{fact}

Adhesive categories have further nice properties:
\begin{fact}
Monomorphisms are stable under pushouts \cite[Lemma 12]{Lack-Sobocinski:04}.
\end{fact}
\begin{fact}
\label{fact:po-along-monos-are-pb}
Pushouts along monomorphisms are pullbacks \cite[Lemma 13]{Lack-Sobocinski:04}.
\end{fact}
\begin{fact}
Pushout complements of monomorphisms (if they exist) are unique up to isomorphism \cite[Lemma 15]{Lack-Sobocinski:04}.
\end{fact}
\begin{fact}
\label{fact:po-pb-decomposition}
Consider the commuting diagram 
\[
\begin{tikzcd}
A \arrow[r, "f"] \arrow[d, "m_1"]
& B \arrow[r, "m_2"] \arrow[d, "m_3"]
& C \arrow[d, "m_4"] \\
D \arrow[r, "g"]
& E \arrow[r, "m_5"]
& F
\end{tikzcd}
\]
where $m_1,\ldots,m_5$ are monomorphisms.
If the outer square is a pushout and the right square is a pullback, then the left square is a pushout. Consequently, all squares are pullback and pushout \cite[Lemma 16]{Lack-Sobocinski:04}.
\end{fact}

Many categories of graphical structures including the category $\cat{Graphs}$ in particular are shown to be examples of adhesive categories~\cite{Lack-Sobocinski:04,Lack-Sobocinski:05,Garner-Lack:12}.
The axioms of an adhesive category enable to derive useful properties, e.g.,
the pullback of two monomorphisms into the same object, which may be seen as two subobjects, may be considered as their intersection,
the pushout of this pullback as their union.
The use of the intersection symbol $\cap$ and the union symbol
$\cup$ is justified as the pullbacks of monomorphisms and the pushouts of
such pullbacks can be constructed by intersections and unions of
subobjects respectively in many typical adhesive categories like $\cat{Sets}$,
$\cat{Graphs}$, $\cat{Hypergraphs}$, etc.
The pushout of the pullback of two monomorphisms is depicted in the following diagram
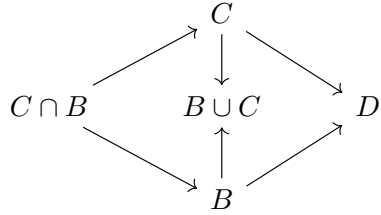
\[
\begin{tikzcd}
& C \arrow[d] \arrow[dr] & \\
C \cap B \arrow[ur] \arrow[dr] & B \cup C & D \\
& B \arrow[u] \arrow[ur] &
\end{tikzcd}
\]
\begin{fact}
\label{fact:mediating-morphism-of-po-of-pb-is-mono}
The mediating morphism of a pushout of a pullback of two monomorphisms with codomain $D$ is a monomorphism into $D$ (cf. \cite[Theorem 5.1]{Lack-Sobocinski:05}).
\end{fact}

One can show associativity of pullbacks.
\begin{fact}
Let $a \colon A \to D, b \colon B \to D, c \colon C \to D$ be monomorphisms.
Then
$(A \cap (B \cap C)) \iso ((A \cap B) \cap C)$.
\end{fact}
Proof sketch: Consider the following diagrams.
\[
\hspace*{-5pt}
\begin{array}{c@{\hspace{0.6cm}}c}
\begin{tikzcd}
A  \arrow[d] & A \cap B \arrow[l] \arrow[d] & (A \cap B) \cap (B\cap C) \arrow[l] \arrow[d]\\
D & B \arrow[l] & B \cap C \arrow[l]
\end{tikzcd}
&
\begin{tikzcd}
C \arrow[d] & C \cap B  \arrow[l] \arrow[d] & (C \cap B) \cap (A\cap B) \arrow[l] \arrow[d]\\
D & B \arrow[l] & A \cap B \arrow[l]
\end{tikzcd}
\end{array}
\]
Using Fact~\ref{fact:composition-lemma-pb} the two outer diagrams are also pullbacks.
Hence, $A \cap (B \cap C) \iso (A \cap B) \cap (B \cap C) \iso C \cap (A \cap B)$.

It is well-known that all the required properties are satisfied by the category of graphs as most of them are shown, e.g., in Ehrig et al.~\cite{Ehrig-Ehrig-Prange.ea:06,Ehrig-Ermel-Golas-Hermann:15}.
They also hold in various further categories the objects of which are graph-like and the basic constructions and properties of pullbacks and pushouts are based on corresponding set-theoretic constructions and properties.
More specifically, they hold in suitable categories of presheaves.

Moreover, adhesive categories have the usual distributivity property of pushouts and pullbacks of monomorphisms \cite[Corollary 5.2]{Lack-Sobocinski:05}.
\begin{fact}
\label{fact:distributivity-lemma}
Let $a\colon A \to D, b\colon B \to D$ and $c\colon C \to D$ be three monomorphisms.
Consider the following two diagrams with three pullbacks and one pushout of monomorphisms in the first diagram and two pullbacks and one pushout of monomorphisms in the second diagram:
\[
\begin{tikzcd}
                           & A \cap C \arrow[r] \arrow[dr] \arrow[d] & C \arrow[dr] & \\
A \cap C \cap B \iso (A \cap C) \cap (A \cap B) \arrow[ur] \arrow[dr] & (A \cap C) \cup (A \cap B) & A \arrow[r]  & D \\
                           & A \cap B \arrow[r] \arrow[ur] \arrow[u] & B \arrow[ur] &
\end{tikzcd}
\]

\[
\begin{tikzcd}
  & A \arrow[dr] & \\
  A \cap (B \cup C) \arrow[dr] \arrow[ur] & B \arrow[r] \arrow[d]  & D \\
  B \cap C \arrow[dr] \arrow[ur] & B \cup C & \\
  & C \arrow[uur] \arrow[u] &
\end{tikzcd}
\]
Then $(A \cap C) \cup (A \cap B) \iso A \cap (B \cup C )$.
\end{fact}

\begin{fact}
In an adhesive category with strict initial object, any morphism with domain $\emptyset$ is a monomorphism (cf. \cite[Lemma 4.1]{Lack-Sobocinski:05}).
\end{fact}

Epi-mono factorizations have a useful property.
\begin{fact}
An epi-mono factorization is unique up to isomorphism (cf. \cite[Theorem 5.1]{Lack-Sobocinski:05}).
This allows to denote the epi-mono factorization of a morphism $f \colon A \to B$ by
$A \xrightarrow[e_{A \to f(A)}]{} f(A) \xrightarrow[m_{f(A) \to B}]{} B$.
\end{fact}

\subsection{The Double-Pushout Approach}
\label{subsec:dpo}
The rewriting formalism for graphs and graph-like structures which we use throughout this paper is the double-pushout (DPO) approach as introduced by Ehrig, Pfender and Schneider in~\cite{Ehrig-Pfender-Schneider73a}.
It was originally introduced for graphs.
However, it is well-defined in adhesive categories.

Let $\cat{C}$ be a category the assumed properties of which are given in Section~\ref{subsec:categorical-notions-and-facts}.
A \emph{rule} is a span of monomorphisms $p = (L \xleftarrow[l]{} K \xrightarrow[r]{} R)$.
$L$ and $R$ are called \emph{left-hand side} and \emph{right-hand side}, $K$ is called \emph{gluing object}.

A rule application to some object $G$ is defined wrt a morphism $g\colon L \to G$ which is called \emph{(left) matching morphism}.
$G$ \emph{directly derives} $H$ if the span $L \xleftarrow[l]{} K \xrightarrow[r]{} R$ and matching morphism $g$ extend to the diagram
\[
\begin{tikzcd}[column sep=large]
  L \arrow[d, "g" left] \arrow[dr, phantom, "(1)" description] & K \arrow[dr, phantom, "(2)" description] \arrow[d, "z" left] \arrow[l, "l" below] \arrow[r, "r" below] & R \arrow[d, "h" left] \\
  G & Z \arrow[l, "m_{Z \to G}" below] \arrow[r, "m_{Z \to H}" below] & H
\end{tikzcd}
\]
such that both squares are pushouts.
$Z$ is called \emph{intermediate object} and $h$ is called \emph{right matching morphism}.

The application of a rule $p$ to $G$ wrt $g$ is called \emph{direct derivation} and is denoted by $G \dder_p H$
(where $g$ is kept implicit). 
A \emph{derivation} from $G$ to $H$ is a sequence of direct derivations $G_0\dder_{p_1}G_1\dder_{p_2} \cdots \dder_{p_n} G_n$ with $G_0 = G$, $G_n = H$ and $n \ge 0$.
If $p_1,\cdots, p_n \in P$, then the derivation is also denoted by $G\dder_P^n H$.
If the length of the derivation does not matter, we write $G \dder_P^*H$.

\bigskip

In our running example (starting with Example~\ref{ex:color-dual}),
we consider two sets of rules. One set is related to the coloring of graphs and
the other one to the construction of the dual graph. As the coloring and the
dual graph are defined in a logic- and set-based way (as NP-problems and their
reductions in general), we recall an explicit construction of the double
pushout (see, e.g., Ehrig \cite{Ehrig:79}).
In this way, the relation of our sample sets of rules to colorings and dual graphs becomes more transparent.

In the category $\cat{Graphs}$
a \emph{rule}  $p = (L\supseteq K\subseteq R)$ consists, without loss of generality, of three graphs $L,K,R \in \G_{\Sigma}$ such that $K$ is a subgraph of $L$ and $R$.

The application of $p$ to a graph $G$ can be explicitly constructed by the following three steps.
\begin{enumerate}
\item
Choose a matching morphism $g\colon L\to G$.
\item
Remove the vertices of $g_V(V_L) \setminus g_V(V_K)$ and the edges of $g_E(E_L) \setminus g_E(E_K)$ yielding $Z$, i.e., $Z=G-(g(L)-g(K))$.
\item
Add $V_R\setminus V_K$ disjointly to $V_Z$ and $E_R\setminus E_K$ disjointly to $E_Z$ yielding (up to isomorphism) the graph $H=Z+(R-K)$ where all edges keep their sources, targets and labels except for $s_H(e)=g_V(s_R(e))$ for $e \in E_R\setminus E_K$ with $s_R(e) \in V_K$ and $t_H(e)=g_V(t_R(e))$ for $e \in E_R\setminus E_K$ with $t_R(e) \in V_K$.
\end{enumerate}
The construction is subject to the \emph{dangling condition} and the \emph{identification condition}.
The first condition is that an edge the source or target of which is in $g(L)-g(K)$ is also in $g(L)-g(K)$.
This ensures that $Z$ becomes a subgraph of $G$ so that $H$ becomes a graph automatically
The second condition is that if different items of $L$ are mapped to the same item in $G$, then they are items of $K$.

The construction yields an intermediate morphism $z\colon K \to Z$ and a right matching morphism $h\colon R \to H$ given by
$z(x) = g(x)$ and $h(x) = z(x)$ for $x \in K$
and $h(x) = in_{R-K}(x)$ for $x \in R - K$
where $in_{R-K} \colon R - K \to H = Z + (R - K)$ is the injection associated to the disjoint union.
Then the diagrams
\[
\begin{tikzcd}[column sep=large]
  L \arrow[d, "g" left] \arrow[dr, phantom, "(1)" description] & K \arrow[dr, phantom, "(2)" description] \arrow[d, "z" left] \arrow[l, phantom, "\supseteq"] \arrow[r, phantom, "\subseteq"] & R \arrow[d, "h" left] \\
  G & Z \arrow[l, phantom, "\supseteq"] \arrow[r, phantom, "\subseteq"] & H
\end{tikzcd}
\]
form a double pushout where the inclusion symbols represent the inclusion morphisms. The diagram (1) is only a
pushout if the matching morphism $g$ satisfies the identification condition.

The right matching morphism $h$ satisfies obviously the dangling and identification conditions with respect to the inverted rule
$p^{-1}=(R \supseteq K \subseteq L)$ such that $p^{-1}$ can be applied to $H$ yielding $G$ with intermediate graph~$Z$.

\begin{exa}\label{ex:color-dual}
This first part of our running example provides the rule sets $P_{\colr}$ and $P_{\dual}$ and a sample derivation for each of the two sets.
$P_{color}$ contains the following rules for some label $a$ and $i=1,...,k$ for some $k \in \N$:
\begin{center}
\begin{tikzpicture}
\node (al) {$\addloop\colon$};
\node (all) [state, right=90pt of al.west, anchor=west] {};
\node (alsup) [right=of all] {$\supseteq$};
\node (alk) [state, right=of alsup] {};
\node (alsub) [right=of alk] {$\subseteq$};
\node (alr) [state, right=of alsub] {}
	edge [loop above] node [above] {\scriptsize $a$} (alr);

\node (ac) [below=27pt of al.west, anchor=west] {$\addcolor(i)\colon$};
\node (acl) [right=89pt of ac.west, anchor=west] {$\emptyset$};
\node (acsup) [right=of acl] {$\supseteq$};
\node (ack) [right=of acsup] {$\emptyset$};
\node (acsub) [right=of ack] {$\subseteq$};
\node (acr) [state, right=of acsub] {}
	edge [loop above] node [above] {\scriptsize $i$} (acr);
	
\node (cc) [below=28pt of ac.west, anchor=west] {$\choosecolr(i)\colon$};
\node (ccl1) [state,right=90pt of cc.west, anchor=west] {}
	edge [loop above] node [above] {\scriptsize $a$} (ccl1);
\node (ccl2) [state,right=of ccl1] {}
	edge [loop above] node [above] {\scriptsize $i$} (ccl2);
\node (ccsup) [right=of ccl2] {$\supseteq$};
\node (cck1) [state,right=of ccsup] {};
\node (cck2) [state,right=of cck1] {}
	edge [loop above] node [above] {\scriptsize $i$} (cck2);
\node (ccsub) [right=of cck2] {$\subseteq$};
\node (ccr1) [state, right=of ccsub] {};
\node (ccr2) [state, right=of ccr1] {}
	edge [loop above] node [above] {\scriptsize $i$} (ccr2) edge [-] (ccr1);
\end{tikzpicture}
\end{center}
The rules can be used to check the $k$-colorability of an unlabeled and undirected graph by constructing the following derivations: 
Apply $\addloop$ to each vertex of the start graph, $\addcolor(i)$ once for each $i = 1,\ldots,k$ and $\choosecolr(i)$ for $i = 1,\ldots,k$ as long as possible.
A derivation of this kind represents a $k$-coloring of the start graph if the derived graph does not contain any of the triangles 
\begin{tikzpicture}[baseline=-3pt]
\node (a) [state] {};
\node (b) [state, above right=7pt and 13pt of a] {}
	edge [loop right] node [right] {\scriptsize $i$} (b)
	edge [-] (a);
\node (c) [state, below right=7pt and 13pt of a] {}
	edge [-] (a)
	edge [-] (b);
\node (space) [below=4pt of c] {};
\end{tikzpicture}
for $i=1,\ldots,k$.
A sample derivation of this kind is
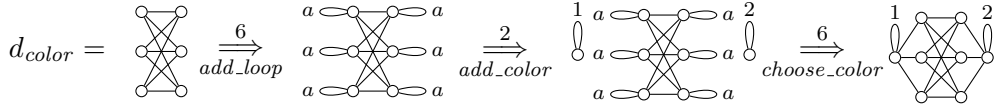
\[
d_{\colr} = \mbox{
\newcommand{\ens}{15pt}

\begin{tikzpicture}[baseline=-2pt]
\node (g1) {
\begin{tikzpicture}
\node (g1n1) [state] {};
\node (g1n2) [state,right=\ens of g1n1] {}
	edge [-] (g1n1);
\node (g1n3) [state,below=\ens of g1n1] {}
	edge [-] (g1n2);
\node (g1n4) [state,right=\ens of g1n3] {}
	edge [-] (g1n1)
	edge [-] (g1n3);
\node (g1n5) [state,below=\ens of g1n3] {}
	edge [-] (g1n2)
	edge [-] (g1n4);
\node (g1n6) [state,right=\ens of g1n5] {}
	edge [-] (g1n1)
	edge [-] (g1n3)
	edge [-] (g1n5);
\end{tikzpicture}
};
\node (g1g2) [right=30pt of g1] {$\dder^{6}_{\addloop}$};
\node (g2) [right=50pt of g1g2] {
\begin{tikzpicture}
\node (g1n1) [state] {}
	edge [loop left] node [left] {\scriptsize $a$} (g1n1);
\node (g1n2) [state,right=\ens of g1n1] {}
	edge [loop right] node [right] {\scriptsize $a$} (g1n2)
	edge [-] (g1n1);
\node (g1n3) [state,below=\ens of g1n1] {}
	edge [loop left] node [left] {\scriptsize $a$} (g1n3)
	edge [-] (g1n2);
\node (g1n4) [state,right=\ens of g1n3] {}
	edge [loop right] node [right] {\scriptsize $a$} (g1n4)
	edge [-] (g1n1)
	edge [-] (g1n3);
\node (g1n5) [state,below=\ens of g1n3] {}
	edge [loop left] node [left] {\scriptsize $a$} (g1n5)
	edge [-] (g1n2)
	edge [-] (g1n4);
\node (g1n6) [state,right=\ens of g1n5] {}
	edge [loop right] node [right] {\scriptsize $a$} (g1n6)
	edge [-] (g1n1)
	edge [-] (g1n3)
	edge [-] (g1n5);
\end{tikzpicture}
};
\node (g2g3) [right=50pt of g2] {$\dder^{2}_{\addcolor}$};
\node (g3) [right=60pt of g2g3] {
\begin{tikzpicture}
\node (g1n1) [state] {}
	edge [loop left] node [left] {\scriptsize $a$} (g1n1);
\node (g1n2) [state,right=\ens of g1n1] {}
	edge [loop right] node [right] {\scriptsize $a$} (g1n2)
	edge [-] (g1n1);
\node (g1n3) [state,below=\ens of g1n1] {}
	edge [loop left] node [left] {\scriptsize $a$} (g1n3)
	edge [-] (g1n2);
\node (g1n4) [state,right=\ens of g1n3] {}
	edge [loop right] node [right] {\scriptsize $a$} (g1n4)
	edge [-] (g1n1)
	edge [-] (g1n3);
\node (g1n5) [state,below=\ens of g1n3] {}
	edge [loop left] node [left] {\scriptsize $a$} (g1n5)
	edge [-] (g1n2)
	edge [-] (g1n4);
\node (g1n6) [state,right=\ens of g1n5] {}
	edge [loop right] node [right] {\scriptsize $a$} (g1n6)
	edge [-] (g1n1)
	edge [-] (g1n3)
	edge [-] (g1n5);
\node (g1n7) [state,left=25pt of g1n3] {}
	edge [loop above] node [above] {\scriptsize $1$} (g1n7);
\node (g1n8) [state,right=25pt of g1n4] {}
	edge [loop above] node [above] {\scriptsize $2$} (g1n8);
\end{tikzpicture}
};
\node (g3g4) [right=60pt of g3] {$\dder^{6}_{\choosecolr}$};
\node (g4) [right=45pt of g3g4] {
\begin{tikzpicture}
\node (g1n1) [state] {};
\node (g1n2) [state,right=\ens of g1n1] {}
	edge [-] (g1n1);
\node (g1n3) [state,below=\ens of g1n1] {}
	edge [-] (g1n2);
\node (g1n4) [state,right=\ens of g1n3] {}
	edge [-] (g1n1)
	edge [-] (g1n3);
\node (g1n5) [state,below=\ens of g1n3] {}
	edge [-] (g1n2)
	edge [-] (g1n4);
\node (g1n6) [state,right=\ens of g1n5] {}
	edge [-] (g1n1)
	edge [-] (g1n3)
	edge [-] (g1n5);
\node (g1n7) [state,left=10pt of g1n3] {}
	edge [loop above] node [above] {\scriptsize $1$} (g1n7)
	edge [-] (g1n1)
	edge [-] (g1n3)
	edge [-] (g1n5);
\node (g1n8) [state,right=10pt of g1n4] {}
	edge [loop above] node [above] {\scriptsize $2$} (g1n8)
	edge [-] (g1n2)
	edge [-] (g1n4)
	edge [-] (g1n6);
\end{tikzpicture}
};
\end{tikzpicture}
}
\]
The start graph is the complete bipartite graph $K_{3,3}$.
The derived graph represents a $2$-coloring of $K_{3,3}$.

$P_{\dual}$ contains the following rules for some label $b$:
\begin{center}
\begin{tikzpicture}
\node (de) {$\doubleedge\colon$};
\node (del1) [state, right=70pt of de.west, anchor=west] {};
\node (del2) [state, right=of del1] {} edge [-] (del1);
\node (desup) [right=of del2] {$\supseteq$};
\node (dek1) [state, right=of desup] {};
\node (dek2) [state, right=of dek1] {} edge [-] (dek1);
\node (desub) [right=of dek2] {$\subseteq$};
\node (der1) [state, right=of desub] {};
\node (der2) [state, right=of der1] {}
	edge [-] (der1)
	edge [-, bend right] node[above]{\scriptsize $b$} (der1);

\node (ae) [below=18pt of de.west, anchor=west] {$\addedge\colon$};
\node (ael1) [state,right=70pt of ae.west, anchor=west] {};
\node (ael2) [state,right=of ael1] {};
\node (aesup) [right=of ael2] {$\supseteq$};
\node (aek1) [state,right=of aesup] {};
\node (aek2) [state,right=of aek1] {};
\node (aesub) [right=of aek2] {$\subseteq$};
\node (aer1) [state,right=of aesub] {};
\node (aer2) [state,right=of aer1] {} 
	edge [-] (aer1);
	
\node (rp) [below=18pt of ae.west, anchor=west] {$\removepair\colon$};
\node (rpl1) [state,right=70pt of rp.west, anchor=west] {};
\node (rpl2) [state,right=of rpl1] {}
	edge [-] (rpl1)
	edge [-, bend right] node [above] {\scriptsize $b$} (rpl1);
\node (rpsup) [right=of rpl2] {$\supseteq$};
\node (rpk1) [state,right=of rpsup] {};
\node (rpk2) [state,right=of rpk1] {};
\node (rpsub) [right=of rpk2] {$\subseteq$};
\node (rpr1) [state, right=of rpsub] {};
\node (rpr2) [state, right=of rpr1] {};
\end{tikzpicture}
\end{center}
The rules can be used to construct the dual graph of an unlabeled, simple, loop-free, undirected graph by constructing the following derivation:
apply $\doubleedge$ for each edge once, followed by $\addedge$ for each pair of vertices that are not connected and finally $\removepair$ as long as possible. 
A sample derivation of this kind is
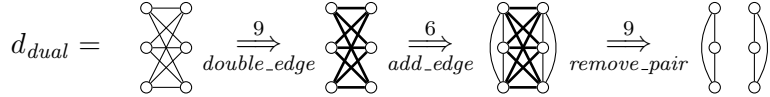
\[
d_{\dual}= \mbox{
\newcommand{\ens}{15pt}
\newcommand{\ensh}{30pt}

\begin{tikzpicture}[baseline=-2pt]
\node (g1) {
\begin{tikzpicture}
\node (n1) [state] {};
\node (n2) [state,right=\ens of n1] {}
	edge [-] (n1);
\node (n3) [state,below=\ens of n1] {}
	edge [-] (n2);
\node (n4) [state,right=\ens of n3] {}
	edge [-] (n1)
	edge [-] (n3);
\node (n5) [state,below=\ens of n3] {}
	edge [-] (n2)
	edge [-] (n4);
\node (n6) [state,right=\ens of n5] {}
	edge [-] (n1)
	edge [-] (n3)
	edge [-] (n5);
\end{tikzpicture}
};
\node (g1g2) [right=35pt of g1] {$\dder^{9}_{\doubleedge}$};
\node (g2) [right=35pt of g1g2] {
\begin{tikzpicture}[
thickedge/.style={-, line width=1pt}
]
\node (n1) [state] {};
\node (n2) [state,right=\ens of n1] {}
	edge [thickedge] (n1);
\node (n3) [state,below=\ens of n1] {}
	edge [thickedge] (n2);
\node (n4) [state,right=\ens of n3] {}
	edge [thickedge] (n1)
	edge [thickedge] (n3);
\node (n5) [state,below=\ens of n3] {}
	edge [thickedge] (n2)
	edge [thickedge] (n4);
\node (n6) [state,right=\ens of n5] {}
	edge [thickedge] (n1)
	edge [thickedge] (n3)
	edge [thickedge] (n5);
\end{tikzpicture}
};
\node (g2g3) [right=30pt of g2] {$\dder^{6}_{\addedge}$};
\node (g3) [right=35pt of g2g3] {
\begin{tikzpicture}[
thickedge/.style={-, line width=1pt}
]
\node (n1) [state] {};
\node (n2) [state,right=\ens of n1] {}
	edge [thickedge] (n1);
\node (n3) [state,below=\ens of n1] {}
	edge [thickedge] (n2)
	edge [-] (n1);
\node (n4) [state,right=\ens of n3] {}
	edge [thickedge] (n1)
	edge [thickedge] (n3)
	edge [-] (n2);
\node (n5) [state,below=\ens of n3] {}
	edge [thickedge] (n2)
	edge [thickedge] (n4)
	edge [-] (n3)
	edge [-, bend left] (n1);
\node (n6) [state,right=\ens of n5] {}
	edge [thickedge] (n1)
	edge [thickedge] (n3)
	edge [thickedge] (n5)
	edge [-] (n4)
	edge [-, bend right] (n2);
\end{tikzpicture}
};
\node (g3g4) [right=40pt of g3] {$\dder^{9}_{\removepair}$};
\node (g4) [right=40pt of g3g4] {
\begin{tikzpicture}
\node (n1) [state] {};
\node (n2) [state,right=\ens of n1] {};
\node (n3) [state,below=\ens of n1] {}
	edge [-] (n1);
\node (n4) [state,right=\ens of n3] {}
	edge [-] (n2);
\node (n5) [state,below=\ens of n3] {}
	edge [-] (n3)
	edge [-,bend left] (n1);
\node (n6) [state,right=\ens of n5] {}
	edge [-] (n4)
	edge [-,bend right] (n2);
\end{tikzpicture}
};
\end{tikzpicture}
}
\]
where each thick line represents a pair of edges labeled by $b$ and $*$ respectively. 
The start graph is again the complete bipartite graph $K_{3,3}$ and the derived graph is the dual graph of $K_{3,3}$.
\end{exa}

\subsection{Local Church-Rosser Properties}
Let $p = (L \xleftarrow[l]{} K \xrightarrow[r]{} R)$ and
$\overline{p} = (\overline{L} \xleftarrow[\overline{l}]{} \overline{K} \xrightarrow[\overline{r}]{} \overline{R})$
be two rules.
\begin{enumerate}
\item
Two direct derivations $G \dder_{p} H$ and $G \dder_{\overline{p}} \overline{H}$ with matching morphisms $g\colon L \to G$ and $\overline{g} \colon \overline{L} \to G$ are \emph{parallel independent}
if there exist two morphisms $f\colon L \to \overline{Z}$ and $\overline{f}\colon \overline{L} \to Z$ such that $g = m_{\overline{Z} \to G} \circ f$ and $\overline{g} = m_{Z \to G} \circ \overline{f}$.
The situation is depicted in the following diagram.
\[
\begin{tikzcd}[column sep = large]
  R \arrow[d, "h" left] & K \arrow[d, "z" left] \arrow[l, "r" below] \arrow[r, "l" below] & L \arrow[dr, "g"] \arrow[drrr, bend left = 20, "f" very near end]
  & & 
  \overline{L} \arrow[dl, "\overline{g}"] \arrow[dlll, bend right = 20, "\overline{f}" very near end] & \overline{K} \arrow[d, "\overline{z}" left] \arrow[l, "\overline{l}" below] \arrow[r, "\overline{r}" below] & \overline{R} \arrow[d, "\overline{h}" left] \\
  H & Z \arrow[l, "m_{Z \to H}" below] \arrow[rr, "m_{Z \to G}" below] & & G & & \overline{Z} \arrow[ll, "m_{\overline{Z} \to G}" below] \arrow[r, "m_{\overline{Z} \to \overline{H}}" below] & \overline{H}
\end{tikzcd}
\]

\item
Successive direct derivations $G \dder_{\overline{p}} \overline{H} \dder_p X$ with the right matching morphism $\overline{h} \colon \overline{R} \to \overline{H}$ and the (left) matching morphism $g' \colon L \to \overline{H}$ are \emph{sequentially independent}
if there exist morphisms $f'\colon \overline{R} \to Z'$ and $\overline{f}\colon L \to \overline{Z}$ such that $\overline{h} = m_{Z' \to \overline{H}} \circ f'$ and $g' = m_{\overline{Z} \to \overline{H}} \circ \overline{f}$.
The situation is depicted in the following diagram.
\[
\begin{tikzcd}[column sep = large]
  \overline{L} \arrow[d, "\overline{g}" left] & \overline{K} \arrow[d, "\overline{z}" left] \arrow[l, "\overline{l}" below] \arrow[r, "\overline{r}" below] & \overline{R} \arrow[dr, "\overline{h}"] \arrow[drrr, bend left = 20, "f'" very near end]
  & & 
  L \arrow[dl, "g'"] \arrow[dlll, bend right = 20, "\overline{f}" very near end] & K \arrow[d, "z'" left] \arrow[l, "l" below] \arrow[r, "r" below] & R \arrow[d, "h'" left] \\
  G & \overline{Z} \arrow[l, "m_{\overline{Z} \to G}" below] \arrow[rr, "m_{\overline{Z} \to \overline{H}}" below] & & \overline{H} & & Z' \arrow[ll, "m_{Z' \to \overline{H}}" below] \arrow[r, "m_{Z' \to X}" below] & X
\end{tikzcd}
\]
\end{enumerate}
It is well-known that parallel independence induces the direct derivations $H \dder_{\overline{p}} X$ and $\overline{H} \dder_{p} X$ with matching morphisms
$m_{Z \to H} \circ \overline{f}$ and $m_{\overline{Z} \to \overline{H}} \circ f$
and that sequential independence induces the derivation $G \dder_p H \dder_{\overline{p}} X$.
The two constructions are called \emph{conflux} and \emph{interchange} respectively (cf. Figure \ref{fig:conflux-interchange} in the Introduction).
The proofs of these two local Church-Rosser theorems in
the context of adhesive categories can be found in Lack and Sobocinski~\cite{Lack-Sobocinski:05}.
The underlying construction of \emph{conflux} is part of the proof of
Lemma~\ref{lemma:moving-preserves-spine-one-step}(\ref{item:moving-preserves-spine-one-step-item-one}).
\newpage

\section{Moving derivations along derivations using independence}
\label{sec:moving-derivations}
In this section, we introduce two types of moving derivations along derivations: 
\emph{forward moving} in \ref{subsec:moveforward} and \emph{backward moving} in \ref{subsec:movebackward}. 
The forward case gets two derivations with the same start object, and one derivation is moved to the derived object of the other derivation (cf. left part of Figure~\ref{fig:moving} in the Introduction). 
The backward case gets the sequential composition of two derivations, and the second derivation is moved to the start object of the first derivation (cf. right part of Figure~\ref{fig:moving}). 
The forward construction is based on parallel independence and iterated application of $\conflux$ while the backward construction is based on sequential independence and iterated application of $\interchange$.

\begin{defi}
Let $(P,\overline{P})$ be a pair of sets of rules.
\begin{enumerate}
\item Then $(P,\overline{P})$ is called \emph{parallel independent} if each two direct derivations of the form $G \dder_p H$ and $G \dder_{\overline{p}} \overline{H}$ with $p \in P$ and $\overline{p} \in \overline{P}$ are parallel independent,
\item and $(P,\overline{P})$ is called \emph{sequentially independent} if each derivation of the form $G \dder_{\overline{p}} \overline{H} \dder_p X$ with $p \in P$ and $\overline{p} \in \overline{P}$ is sequentially independent.
\end{enumerate}
\end{defi}

\begin{exa}\label{ex:independence-color-dual}
The pair of sets of rules $(P_{\colr},P_{\dual})$ is parallel and sequentially independent.
A $P_{\colr}$-rule application and a $P_{\dual}$-rule application to the same object would not be parallel independent if one removes something while the other accesses it.
This does not happen.
The only removing rules are $\choosecolr$ and $\removepair$.
If the rule $\choosecolr$ is applied, then an $a$-loop is removed, but $a$-loops are not accessible by $P_{\dual}$-rules.
If $\removepair$ is applied, then a $b$-edge and $*$-edge are removed, but none of these edges can be accessed by a $P_{\colr}$-rule.
An application of a $P_{\dual}$-rule followed by an application of a $P_{\colr}$-rule would not be sequentially independent if the first generates something that the second accesses or if the right matching morphism of the first step accesses something that is removed by the second step.
This does not happen.
The rule $\doubleedge$ adds a $b$-edge, the rule $\addedge$ a $*$-edge, but none of them can be accessed by a $P_{\colr}$-rule.
And the only removing $P_{\colr}$-rule is $\choosecolr$ removing an $a$-loop, but no right-hand side of a $P_{\dual}$-rule contains such a loop.
Summarizing the statement holds.

It may be noticed that this is not the case for the pair $(P_{\dual},P_{\colr})$.
Clearly it is parallel independent as the order of the rule application is not significant in this case. 
But if one applies first the rule $\choosecolr$ generating an edge and then $\doubleedge$ to this newly generated edge, then the sequential independence fails.
\end{exa}

\subsection{Moving forward using parallel independence}\label{subsec:moveforward}

Two parallel independent rule applications do not interfere with each other.
Each one keeps the matching morphism of the other one intact
meaning that the matching morphism of each
rule application can be reconstructed as matching morphism in the derived object of the other rule application.
Consequently, it can be applied after the other one yielding the same result in both cases (cf. left part of Figure \ref{fig:conflux-interchange}). 
A possible interpretation of this strong confluence property is that the application of $p$ is moved along the application of $\overline{p}$ and – symmetrically – the application of $\overline{p}$ along the application of $p$.
This carries over to derivations.
\newpage

\begin{prop}\label{prop:moving}
Let $(P,\overline{P})$ be a parallel independent pair of sets of rules.
Let $d = (G \dder^n_P H)$ and $\overline{d} = (G \dder^m_{\overline{P}} \overline{H})$ be two derivations.
Then the iteration of \emph{conflux} as long as possible yields two derivations of the form $H \dder^m_{\overline{P}} X$ and $\overline{H} \dder^n_P X$ for some object $X$.
\end{prop}

\begin{figure}[t]
\begin{subfigure}{0.4\textwidth}
\centering
\begin{tikzcd}
& G \arrow[r, Rightarrow, "\overline{p}_1" below] \arrow[d, Rightarrow, "P" right, "0" left]
& \overline{G}_1 \arrow[d, Rightarrow, "P" right, "0" left] & \\
H \arrow[r, phantom, "=" description] &  G \arrow[r, Rightarrow, "\overline{p}_1" below]
& \overline{G}_1 & X \arrow[l, phantom, "=" description]
\end{tikzcd}
\caption{$n=0, m=1$}
\label{fig:moving-n=0,m=1}
\end{subfigure}
\begin{subfigure}{0.55\textwidth}
\centering
\begin{tikzcd}
& G \arrow[r, Rightarrow, "\overline{p}_1" below] \arrow[d, Rightarrow, "p_1" right]
  \arrow[dr, phantom, "(1)"]
& \overline{G}_1 \arrow[d, Rightarrow, "p_1" right] & \\
& G_1 \arrow[r, Rightarrow, "\overline{p}_1" below] \arrow[d, Rightarrow, "P" right, "n" left]
\arrow[dr, phantom, "(2)"]
& \overline{X}_1 \arrow[d, Rightarrow, "P" right, "n" left] & \\
H \arrow[r, phantom, "=" description] &
G_{n+1} \arrow[r, Rightarrow, "\overline{p}_1" below]
& \overline{X}_{n+1} & X \arrow[l, phantom, "=" description]
\end{tikzcd}
\caption{$n+1, m=1$}
\label{fig:moving-n+1,m=1}
\end{subfigure}
\begin{subfigure}{0.4\textwidth}
\centering
\begin{tikzcd}
G \arrow[r, Rightarrow, "\overline{P}" below, "0" above] \arrow[d, Rightarrow, "P" right, "n" left]
& G \arrow[d, Rightarrow, "P" right, "n" left]
& \overline{H} \arrow[l, phantom, "=" description]
\\
H \arrow[r, Rightarrow, "\overline{P}" below, "0" above]
& H & X \arrow[l, phantom, "=" description]
\end{tikzcd}
\caption{$m=0$, general case}
\label{fig:moving-m=0-general}
\end{subfigure}
\begin{subfigure}{0.55\textwidth}
\centering
\begin{tikzcd}
 G \arrow[r, Rightarrow, "\overline{p}_1" below] \arrow[d, Rightarrow, "P" right, "n" left]
  \arrow[dr, phantom, "(3)"]
  & \overline{G}_1 \arrow[r, Rightarrow, "\overline{P}" below, "m" above] \arrow[d, Rightarrow, "P" right, "n" left]
  \arrow[dr, phantom, "(4)"]
  & \overline{G}_{m+1} \arrow[d, Rightarrow, "P" right, "n" left]
  & \overline{H} \arrow[l, phantom, "=" description] \\ 
H \arrow[r, Rightarrow, "\overline{p}_1" below]
& X_1 \arrow[r, Rightarrow, "\overline{P}" below, "m" above]
& X_{m+1}
& X \arrow[l, phantom, "=" description]
\end{tikzcd}
\caption{$m+1$, general case}
\label{fig:moving-m+1-general}
\end{subfigure}
\caption{Diagrams used in the proof of Proposition~\ref{prop:moving}}
\end{figure}

\begin{proof}
Formally, the construction is done inductively on the lengths of derivations first for $m=1$ and then in general:

Base $n=0$ for $m=1$:
By definition, one can assume that a $0$-derivation does not change the start object so that 
the direct derivation $G \dder_{\overline{p}_1} \overline{G}_1$ and the $0$-derivation $\overline{G}_1 \dder^0_P \overline{G}_1$ fulfill the statement (cf. diagram in Figure~\ref{fig:moving-n=0,m=1}).

Step $n+1$ for $m=1$:
The derivation $d$ can be decomposed into $G \dder_{p_1} G_1 \dder^n_{P} G_{n+1}$.
Then the operation $\conflux$ can be applied to $G \dder_{p_1} G_1$ and $G \dder_{\overline{p}_1} \overline{G}_1$ yielding $G_1 \dder_{\overline{p}_1} \overline{X}_1$ and
$\overline{G}_1 \dder_{p_1} \overline{X}_1$ as depicted in subdiagram (1) of the diagram in Figure~\ref{fig:moving-n+1,m=1}.
Now one can apply the induction hypothesis to $G_1 \dder_{\overline{p}_1} \overline{X}_1$ and $G_1 \dder^n_{P} G_{n+1}$ yielding
$\overline{X}_1 \dder^n_{P} \overline{X}_{n+1}$ and $G_{n+1} \dder_{\overline{p}_1} \overline{X}_{n+1}$ as depicted in subdiagram (2).
The composition of the diagrams (1) and (2) completes the construction.

Base $m=0$ for the general case:
Using the same argument as for $n=0$, one gets the required derivations as the given $G \dder^n_P H$ and the $0$-derivation $H \dder^0_{\overline{P}} H$ (cf. diagram in Figure~\ref{fig:moving-m=0-general}).

Step $m+1$ for the general case:
The derivation $\overline{d}$ can be decomposed into $G \dder_{\overline{p}_1}$\linebreak[4]$\overline{G}_1 \dder^m_{\overline{P}} \overline{G}_{m+1}$.
The construction for $m=1$ can be applied to $G \dder_{\overline{p}_1} \overline{G}_1$ and $G \dder^n_{P} H$ yielding
$\overline{G}_1 \dder^n_{P} X_1$ and $H \dder_{\overline{p}_1} X_1$ as depicted in subdiagram (3) of the diagram in Figure~\ref{fig:moving-m+1-general}.
Now one can apply the induction hypothesis to $\overline{G}_1 \dder^m_{\overline{P}} \overline{G}_{m+1}$ and $\overline{G}_1 \dder^n_{P} X_1$ yielding
$\overline{G}_{m+1} \dder^n_{P} X_{m+1}$ and $X_1 \dder^m_{\overline{P}} X_{m+1}$ as depicted in subdiagram (4) .
The composition of the diagrams (3) and (4) completes the construction.

The construction yields a grid of \emph{conflux}-squares.
It is easy to see by induction on the number of \emph{conflux}-applications
that every sequence of \emph{conflux}-applications of length $m\times n$ yields
the same grid up to isomorphism. One starts from the two
given derivations as induction base with $0$ \emph{conflux}-applications.
They form a substructure of the grid.
By induction hypothesis one can assume that $k$ \emph{conflux}-applications yield a substructure of the grid. Then every possible
further \emph{conflux}-application as induction step stays within the grid as \emph{conflux} is unique up to isomorphism so that one can take the square of the grid.
\end{proof}

\begin{rem}
In diagrammatic form, the situation looks as follows:
\[
\begin{tikzcd}
G \arrow[r, Rightarrow, "\overline{P}" below, "m" above] \arrow[d, Rightarrow, "P" right, "n" left]
& \overline{H} \\
H
& 
\end{tikzcd}
~~~~~~
\leadsto
~~~~~~
\begin{tikzcd}
G \arrow[r, Rightarrow, "\overline{P}" below, "m" above] \arrow[d, Rightarrow, "P" right, "n" left]
& \overline{H} \arrow[d, Rightarrow, "P" right, "n" left] \\
H \arrow[r, Rightarrow, "\overline{P}" below, "m" above]
& X
\end{tikzcd}
\]
The situation of Proposition \ref{prop:moving} can be seen as moving $d$ along $\overline{d}$ and $\overline{d}$ along $d$.
The constructed derivation $\overline{H} \dder^*_P X$ is called \emph{moved variant} of $d$ and denoted by $\move(d,\overline{d})$, and the second constructed derivation $H \dder^*_{\overline{P}} X$ is called moved variant of $\overline{d}$ and denoted by $\move(\overline{d}, d)$.
\end{rem}

\begin{exa}\label{ex:moved-variant-color-dual}
As pointed out in Example \ref{ex:independence-color-dual}, $(P_{\colr},P_{\dual})$ is a parallel independent pair of sets of rules.
Therefore, the derivation $d_{\colr}$ of Example \ref{ex:color-dual} can be moved along the derivation $d_{\dual}$ of Example \ref{ex:color-dual} yielding the moved variant
\[
\move(d_{\colr},d_{\dual})=
\mbox{
\newcommand{\ens}{15pt}
\begin{tikzpicture}[baseline=-2pt]
\node (g1) {
\begin{tikzpicture}
\node (n1) [state] {};
\node (n2) [state,right=\ens of n1] {};
\node (n3) [state,below=\ens of n1] {}
	edge [-] (n1);
\node (n4) [state,right=\ens of n3] {}
	edge [-] (n2);
\node (n5) [state,below=\ens of n3] {}
	edge [-] (n3)
	edge [-,bend right] (n1);
\node (n6) [state,right=\ens of n5] {}
	edge [-] (n4)
	edge [-,bend left] (n2);
\end{tikzpicture}
};
\node (g1g2) [right=30pt of g1] {$\dder^{6}_{\addloop}$};
\node (g2) [right=45pt of g1g2] {
\begin{tikzpicture}
\node (n1) [state] {}
	edge [loop left] node [left] {\scriptsize $a$} (n1);
\node (n2) [state,right=\ens of n1] {}
	edge [loop right] node [right] {\scriptsize $a$} (n2);
\node (n3) [state,below=\ens of n1] {}
	edge [loop left] node [left] {\scriptsize $a$} (n3)
	edge [-] (n1);
\node (n4) [state,right=\ens of n3] {}
	edge [loop right] node [right] {\scriptsize $a$} (n4)
	edge [-] (n2);
\node (n5) [state,below=\ens of n3] {}
	edge [loop left] node [left] {\scriptsize $a$} (n5)
	edge [-] (n3)
	edge [-,bend right] (n1);
\node (n6) [state,right=\ens of n5] {}
	edge [loop right] node [right] {\scriptsize $a$} (n6)
	edge [-] (n4)
	edge [-,bend left] (n2);
\end{tikzpicture}
};
\node (g2g3) [right=45pt of g2] {$\dder^{2}_{\addcolor}$};
\node (g3) [right=55pt of g2g3] {
\begin{tikzpicture}
\node (n1) [state] {}
	edge [loop left] node [left] {\scriptsize $a$} (n1);
\node (n2) [state,right=\ens of n1] {}
	edge [loop right] node [right] {\scriptsize $a$} (n2);
\node (n3) [state,below=\ens of n1] {}
	edge [loop left] node [left] {\scriptsize $a$} (n3)
	edge [-] (n1);
\node (n4) [state,right=\ens of n3] {}
	edge [loop right] node [right] {\scriptsize $a$} (n4)
	edge [-] (n2);
\node (n5) [state,below=\ens of n3] {}
	edge [loop left] node [left] {\scriptsize $a$} (n5)
	edge [-] (n3)
	edge [-,bend right] (n1);
\node (n6) [state,right=\ens of n5] {}
	edge [loop right] node [right] {\scriptsize $a$} (n6)
	edge [-] (n4)
	edge [-,bend left] (n2);
\node (n7) [state,left=25pt of n3] {}
	edge [loop above] node [above] {\scriptsize $1$} (n7);
\node (n8) [state,right=25pt of n4] {}
	edge [loop above] node [above] {\scriptsize $2$} (n8);
\end{tikzpicture}
};
\node (g3g4) [right=60pt of g3] {$\dder^{6}_{\choosecolr}$};
\node (g4) [right=50pt of g3g4] {
\begin{tikzpicture}
\node (n1) [state] {};
\node (n2) [state,right=\ens of n1] {};
\node (n3) [state,below=\ens of n1] {}
	edge [-] (n1);
\node (n4) [state,right=\ens of n3] {}
	edge [-] (n2);
\node (n5) [state,below=\ens of n3] {}
	edge [-] (n3)
	edge [-,bend right] (n1);
\node (n6) [state,right=\ens of n5] {}
	edge [-] (n4)
	edge [-,bend left] (n2);
\node (n7) [state,left=\ens of n3] {}
	edge [loop above] node [above] {\scriptsize $1$} (n7)
	edge [-] (n1)
	edge [-] (n3)
	edge [-] (n5);
\node (n8) [state,right=\ens of n4] {}
	edge [loop above] node [above] {\scriptsize $2$} (n8)
	edge [-] (n2)
	edge [-] (n4)
	edge [-] (n6);
\end{tikzpicture}
};
\end{tikzpicture}
}
\]
\end{exa}

\subsection{Moving backward using sequential independence}\label{subsec:movebackward}

Backward moving can be constructed analogously to forward moving if one uses sequential independence and the $\interchange$-operation instead of parallel independence and $\conflux$. 
Another possibility is to reduce backward moving to forward moving by means of inverse derivations.

\begin{prop}
\label{prop:moving2}
Let $(P,\overline{P})$ be a sequentially independent pair of sets of rules and $\overline{d} = (G \dder^m_{\overline{P}} \overline{H})$ and $d' = (\overline{H} \dder^n_P X)$ be two derivations.
Then the iteration of \emph{interchange} applied to a $\overline{P}$-rule application followed by a $P$-rule application as long as possible yields two derivations of the form $G \dder^n_P H$ and $H \dder^m_{\overline{P}} X$ for some object $H$.
\end{prop}

\begin{proof}
The sequential independence of $(P,\overline{P})$ implies, obviously, the parallel independence of $(P,(\overline{P})^{-1})$ as right matching morphisms of applications of $\overline{P}$-rules become left matching morphisms of the corresponding inverse rule applications, and an application of $\interchange$ becomes an application of $\conflux$.
Therefore, Proposition \ref{prop:moving} can be applied to $d'$ and $(\overline{d})^{-1}$ (deriving $\overline{H}$ into $G$) yielding $\move(d',(\overline{d})^{-1})$ of the form $G \dder^n_P H$ and $\move((\overline{d})^{-1},d')$ of the form $X \dder^m_{(\overline{P})^{-1}} H$ such that $\move(d',(\overline{d})^{-1})$ and $\move((\overline{d})^{-1},d')^{-1}$ prove the proposition.
\end{proof}

\begin{rem}
In diagrammatic form, the situation of Proposition \ref{prop:moving2} looks as follows:
\[
\begin{tikzcd}
G \arrow[r, Rightarrow, "\overline{P}" below, "m" above]
& \overline{H} \arrow[d, Rightarrow, "P", "n" left]\\
& X
\end{tikzcd}
~~~~~~
\leadsto
~~~~~~
\begin{tikzcd}
G \arrow[r, Rightarrow, "\overline{P}" below, "m" above] \arrow[d, Rightarrow, "P" right, "n" left]
& \overline{H} \arrow[d, Rightarrow, "P" right, "n" left] \\
H \arrow[r, Rightarrow, "\overline{P}" below, "m" above]
& X
\end{tikzcd}
\]
This can be seen as moving $d'$ backward along $\overline{d}$.
The resulting derivation $\move(d',(\overline{d})^{-1})$ is called \emph{backward moved variant} of $d'$ denoted by $\evom(d',\overline{d})$.
\end{rem}

\begin{exa}\label{ex:evom-color-dual}
As $(P_{\colr},P_{\dual})$ is not only parallel independent, but also sequentially independent, one can move the derivation $\move(d_{\colr},d_{\dual})$ backward along $d_{\dual}$ yielding $\evom(\move(d_{\colr},d_{\dual}),d_{\dual}) = d_{\colr}$.
\end{exa}

\section{Accessed Parts, Restrictions and Spines}
\label{sec:accessed-part-and-spine}
In this section, we recall another long-known concept: the restriction
of a derivation with respect to a monomorphism into the start object (cf. \cite{Ehrig-Ehrig-Prange.ea:06,Ehrig-Ermel-Golas-Hermann:15}).
This is always possible if the accessed part factors through this
monomorphism where the accessed part is a monomorphism into the start
object itself that identifies the parts of the start
object that are accessed by the left matching morphism of some derivation step.
In particular, the derivation can be restricted to the accessed part.
We call this particular restriction the spine of the derivation.

The constructions in this section are based on the following assumptions.
\begin{asm}
  \label{assumption:accessed-part}\hfill
    \begin{enumerate}
    \item
    Let $d = (G \dder_P^* X)$ be a derivation with two cases
      $G \dder_P^0 X$ only if $G = X$, and
      $(G \dder_P^{n+1} X) = (G \dder_p H \dder_P^{n} X)$ for $n \in \N$ where the first step is defined by the double pushout
\[
\begin{tikzcd}[column sep=huge]
  L \arrow[d, "g" left] \arrow[dr, phantom, "(1)" description] & K \arrow[dr, phantom, "(2)" description] \arrow[d, "z" left] \arrow[l, "l" below] \arrow[r, "r" below] & R \arrow[d, "h" left] \\
  G & Z \arrow[l, "m_{Z \to G}" below] \arrow[r, "m_{Z \to H}" below] & H
\end{tikzcd}
\]
with $p = (L \xleftarrow[l]{} K \xrightarrow[r]{} R) \in P$.
The remaining derivation is denoted by $\hat{d} = \tail(d) = (H \dder_P^n X)$.

\item
  Let $(L \xrightarrow[g]{} G) = (L \xrightarrow[e_{L \to g(L)}]{} g(L) \xrightarrow[m_{g(L) \to G}]{} G)$ be the epi-mono factorization of $g$.

\item
  Let $G' \xrightarrow[m_{G' \to G}]{} G$ be a monomorphism for some object $G'$.
  \end{enumerate}
\end{asm}

\subsection{Accessed Parts}
We start by constructing the accessed part of a derivation.
In the category of directed edge-labeled graphs, it contains all vertices and edges of the start graph that are accessed by a left matching morphism of some derivation step.
It can be defined as a subgraph of the start graph $G$ by induction on the length of the derivation.

\base
  $\acc(G\dder^0_P X)=\emptyset$, and

\step
  $\acc(G \dder_{p} H\dder^n_P X )=(\acc(H\dder^n_P X) \cap Z)\cup g(L)$
  where
$\acc(H\dder^n_P X) \subseteq H$ is given by induction hypothesis, 
$L$ is the left-hand side of $p$, $g$ is the first matching morphism, and $Z$ is the intermediate graph of the first rule application.

\begin{exa}\label{ex:accessed-color}
Considering $d_{\colr}$ in Example \ref{ex:color-dual}, $\acc(d_{\colr}) = $
{
\newcommand{\ens}{8pt}
\begin{tikzpicture}[baseline=-11pt]
\node (n1) [state] {};
\node (n2) [state,right=\ens of n1] {};
\node (n3) [state,below=\ens of n1] {};
\node (n4) [state,right=\ens of n3] {};
\node (n5) [state,below=\ens of n3] {};
\node (n6) [state,right=\ens of n5] {};
\node (foot) [below=0pt of n5.south, anchor=north] {};
\node (head) [above=0pt of n1.north, anchor=south] {};
\end{tikzpicture}
}
as all vertices of $K_{3,3}$ are accessed by the applications of the rules $\addloop$, $\choosecolr(1)$, and $\choosecolr(2)$, but none of the edges is accessed.
\end{exa}

The definition can be lifted onto the level of adhesive categories by replacing subgraphs by monomorphisms, the empty graph by the strict initial object, the intersection by a pullback of monomorphisms, and the union by a pushout of a pullback of monomorphisms.

\begin{construction}
  \label{construction:accessed-part}
  The \emph{accessed part} $\acc(d)\colon \ACC(d) \to G$ of $d$ is constructed by induction on the length of $d$.  

  \base
    $\acc(G \dder_P^0 X) = (\emptyset_G\colon \emptyset \to G)$, and

  \step
    $\acc( G \dder_P^{n+1} X) = \acc(G \dder_p H \dder_P^n X)\colon \ACC(d) \to G$ is constructed in four steps
    where $\acc(\hat{d}) = \acc(\tail(d))\colon \ACC(\hat{d}) \to H$ is the accessed part of $\hat{d}$ using the induction hypothesis for a derivation of length $n$.

\begin{enumerate}
\item
  \label{item:construction-accessed-part-pb-one}
  Construct the pullback
  \[
  \begin{tikzcd}[column sep=huge]
  Z \cap \ACC(\hat{d}) \arrow[d, "m_Z" left] \arrow[r, "m_{\ACC(\hat{d})}" below] & \ACC(\hat{d}) \arrow[d, "\acc(\hat{d})" left] \\
  Z \arrow[r, "m_{Z \to H}" below] & H
  \end{tikzcd}
\]
\item
  \label{item:construction-accessed-part-pb-two}
  Construct the pullback
\[
  \begin{tikzcd}[column sep=huge]
    g(L) \cap (Z \cap \ACC(\hat{d})) \arrow[dd, "m_{g(L)}"] \arrow[r, "m_{Z \cap \ACC(\hat{d})}" below] & Z \cap \ACC(\hat{d}) \arrow[d, "m_Z" left] \\
      & Z \arrow[d, "m_{Z \to G}" left] \\
  g(L) \arrow[r, "m_{g(L) \to G}" below] & G
  \end{tikzcd}
\]

\item
    \label{item:construction-accessed-part-po}
    Construct the pushout
\[
  \begin{tikzcd}[column sep=huge]
    g(L) \cap (Z \cap \ACC(\hat{d})) \arrow[d, "m_{g(L)}"] \arrow[r, "m_{Z \cap \ACC(\hat{d})}" below] & Z \cap \ACC(\hat{d}) \arrow[d, "\overline{m}_{Z \cap \ACC(\hat{d})}" left] \\
  g(L) \arrow[r, "\overline{m}_{g(L)}" below] & g(L) \cup (Z \cap \ACC(\hat{d}))
  \end{tikzcd}
\]
\item
  \label{item:construction-accessed-part-mediate}
  Using the commutativity of the pullback in Step \ref{item:construction-accessed-part-pb-two}, the pushout in Step \ref{item:construction-accessed-part-po} induces a mediating morphism $\acc(d)\colon \ACC(d) \to G$ with $\ACC(d) = g(L) \cup (Z \cap \ACC(\hat{d}))$ and the commutative diagrams
  \vspace*{10pt}
\[
  \begin{tikzcd}[column sep=huge]
    & Z \cap \ACC(\hat{d}) \arrow[d, "\overline{m}_{Z \cap \ACC(\hat{d})}" left] \arrow[dr, "m_{Z}" left] & \\
    g(L) \arrow[r, "\overline{m}_{g(L)}" above] \arrow[drr, bend right = 15, "m_{g(L) \to G}" above] & \ACC(d) \arrow[dr, "\acc(d)" above] & Z \arrow[d, "m_{Z \to G}"] \\
    & & G
  \end{tikzcd}
\]
\end{enumerate}
\vspace*{10pt}
\noindent Altogether the situation is depicted in the following diagram.
\vspace*{10pt}
\[
\begin{tikzcd}[column sep=huge]
  L \arrow[ddd, controls={+(-3.5,0) and +(-3.5,0)}, "g" left] \arrow[d, "e_{L \to g(L)}"] & K \arrow[l, "l" below] \arrow[r, "r" below] \arrow[ddd, controls={+(3,0) and +(3,0)}, "z" left]  & R \arrow[ddd, bend left = 60, "h" left] \\
  g(L) \arrow[dd, bend right = 60, "m_{g(L) \to G}" above left] \arrow[d, "\overline{m}_{g(L)}" right] & g(L) \cap (Z \cap \ACC(\hat{d})) \arrow[d, "m_{Z \cap \ACC(\hat{d})}" left] \arrow[l, "m_{g(L)}" above] & \\
  \ACC(d) \arrow[d, "\acc(d)"]  &  Z \cap \ACC(\hat{d}) \arrow[d, "m_Z" left]  \arrow[l, "\overline{m}_{Z \cap \ACC(\hat{d})}" above] \arrow[r, "m_{\ACC(\hat{d})}" near end] & \ACC(\hat{d}) \arrow[d, "\acc(\hat{d})" left] \\
  G  & Z \arrow[l, "m_{Z \to G}" below] \arrow[r, "m_{Z \to H}" below] & H
\end{tikzcd}
\]
\end{construction}
\vspace*{10pt}

\begin{rem}
Observe that $acc(d)$ is a monomorphism by adhesivity (cf. Fact~\ref{fact:mediating-morphism-of-po-of-pb-is-mono}).
\end{rem}

\subsection{Restrictions}

In the category of graphs, it is known that the derivation $d$ can be restricted to every subgraph $G'$ of $G$
with $\ACC(d) \subseteq G'$ by clipping off vertices and edges outside of $G'$ (or their isomorphic counterparts) from all graphs of $d$ keeping all matches invariant in this way (cf.~\cite{Ehrig-Ehrig-Prange.ea:06}).
The construction can be generalized to a monomorphism from $G'$ to $G$
if $\acc(d)\colon \ACC(d) \to G$ factors through it on the level of adhesive categories.
The restriction is constructed by induction on the length of the derivation $d$.
To make sure that the construction works, we use
two lemmata where the first lemma concerns the restriction of a direct
derivation, and the second lemma guarantees that the restriction can be
continued after the restriction of the first derivation step.

The Restriction Theorem 6.18 in~\cite{Ehrig-Ehrig-Prange.ea:06} concerns only a single derivation step similarly to the first lemma.
Restriction of arbitrary derivations is defined in~\cite{Ehrig-Ehrig-Prange.ea:06} as inverse to extension without an explicit construction as provided in the following.
\vspace*{10pt}
\begin{lem}
  \label{lemma:restriction}
  Let $g'\colon L \to G'$ be a morphism and $m_{G' \to G}\colon G' \to G$ be a monomorphism for some object $G'$ such that $g = m_{G' \to G} \circ g'$.
  Then the double pushout defining $G \dder_p H$ can be\newpage \noindent decomposed into
  \[
\begin{tikzcd}[column sep=huge]
  L \arrow[dr, phantom, "(4)" description] \arrow[dd, controls={+(-2,0) and +(-2,0)}, "g" left] \arrow[d, "g'"] & K \arrow[dr, phantom, "(5)" description] \arrow[l, "l" below] \arrow[r, "r" below] \arrow[d, "z'"] & R \arrow[dd, controls={+(+2,0) and +(+2,0)}, "h" right] \arrow[d, "h'"]\\
  G' \arrow[dr, phantom, "(3)" description] \arrow[d, "m_{G' \to G}" left] & Z' \arrow[dr, phantom, "(6)" description] \arrow[l, "m_{Z' \to G'}" below] \arrow[r, "m_{Z' \to H'}" below] \arrow[d, "m_{Z' \to Z}" near end] & H' \arrow[d, "m_{H' \to H}"] \\
  G  & Z \arrow[l, "m_{Z \to G}" below] \arrow[r, "m_{Z \to H}" below] & H
\end{tikzcd}
\]
such that $z = m_{Z' \to Z} \circ z'$, $h = m_{H' \to H} \circ h'$, (4) and (5) are pushouts and (3) and (6) are pullbacks.
In particular, (4) and (5) define a direct derivation $G' \dder_p H'$.
\end{lem}

\begin{proof}
  (3) can be constructed as pullback.
  Using the commutativity of (1) in the assumption,
  the universal property of (3) yields the mediating morphism $z'$ such that $z = m_{Z' \to Z} \circ z'$ and (4) commutes.
  As (1) is a pushout and (3) a pullback, the pushout-pullback decomposition property (Fact~\ref{fact:po-pb-decomposition}) implies that (4) is a pushout.
  Then (5) can be constructed as pushout.
  The universal property of pushouts gives a mediating morphism $m_{H' \to H} \colon H' \to H$ satisfying
  $h = m_{H' \to H} \circ h'$ and $m_{H' \to H} \circ m_{Z' \to H'} = m_{Z \to H} \circ m_{Z' \to Z}$.
  Using Fact~\ref{fact:composition-lemma-po}, $(6)$ is a pushout.
  As a pushout of monomorphisms, (6) is a pullback, too.
  Moreover, the pushouts (4) and (5) define a direct derivation $G' \dder_p H'$.
\end{proof}

\begin{defi}
  The derivation $G' \dder_p H'$ given in Lemma~\ref{lemma:restriction} is called \emph{single-step restriction} of $G \dder_p H$ wrt $m_{G' \to G}\colon G' \to G$, denoted by $(G \dder_p H)|m_{G' \to G}$.
\end{defi}

\begin{rem}
  \label{remark:restriction}\hfill
  \begin{enumerate}
  \item
    \label{item:remark-restriction-one}
    According to the assumed epi-mono factorization of $g$ and the commutativity of the lower triangle in Step~\ref{item:construction-accessed-part-mediate} of the Construction~\ref{construction:accessed-part},
    we have
    \[
    g = m_{g(L) \to G} \circ e_{L \to g(L)} = \acc(d) \circ \overline{m}_{g(L)} \circ e_{L \to g(L)}.
    \]
    Therefore, the single-step restriction $(G \dder_p H)|acc(d)$ is defined wrt $\acc(d)\colon \ACC(d) \to G$, in particular.
    Moreover, one gets the following decomposition of the double pushout defining $G \dder_p H$:
  \[
\begin{tikzcd}[column sep=huge]
  L \arrow[dr, phantom, "(\widetilde{4})" description] \arrow[dd, controls={+(-2,0) and +(-2,0)}, "g" left] \arrow[d, "\widetilde{g}"] & K \arrow[dr, phantom, "(\widetilde{5})" description] \arrow[l, "l" below] \arrow[r, "r" below] \arrow[d, "\widetilde{z}"] & R \arrow[dd, controls={+(2,0) and +(2,0)}, "h" right] \arrow[d, "\widetilde{h}"]\\
  \ACC(d) \arrow[dr, phantom, "(\widetilde{3})" description] \arrow[d, "\acc(d)" right] & \widetilde{Z} \arrow[dr, phantom, "(\widetilde{6})" description] \arrow[l, "m_{\widetilde{Z} \to \ACC(d)}" below] \arrow[r, "m_{\widetilde{Z} \to \widetilde{H}}" below] \arrow[d, "m_{\widetilde{Z} \to Z}" near end] & \widetilde{H} \arrow[d, "m_{\widetilde{H} \to H}"] \\
  G  & Z \arrow[l, "m_{Z \to G}" below] \arrow[r, "m_{Z \to H}" below] & H\\
\end{tikzcd}
\]
    where $(\widetilde{3})$ is a pullback, $(\widetilde{4}), (\widetilde{5})$ and $(\widetilde{6})$ are pushouts, and $(\widetilde{6})$ is also pullback.
    
  \item
    \label{item:remark-restriction-two}
    Moreover, if $\acc(d)$ factors through $m_{G' \to G}$, then $(\widetilde{3})$ and $(\widetilde{6})$ decompose into the following pullbacks
      \[
\begin{tikzcd}[column sep=3cm]
  \ACC(d) \arrow[dr, phantom, "(3'')"] \arrow[d, "m_{\ACC(d) \to G'}" right] \arrow[dd, bend right, "\acc(d)" left] & \widetilde{Z} \arrow[dr, phantom, "(6'')"] \arrow[l, "m_{\widetilde{Z} \to \ACC(d)}" below] \arrow[r, "m_{\widetilde{Z} \to \widetilde{H}}" below] \arrow[d, "m_{\widetilde{Z} \to Z'}"] & \widetilde{H} \arrow[d, "m_{\widetilde{H} \to H'}"] \arrow[dd, controls={+(2,0) and +(2,0)}, "m_{\widetilde{H} \to H}" right] \\
    G' \arrow[dr, phantom, "(3')"] \arrow[d, "m_{G' \to G}" right] & Z' \arrow[dr, phantom, "(6')"] \arrow[l, "m_{Z' \to G'}" below] \arrow[r, "m_{Z' \to H'}" below] \arrow[d, "m_{Z' \to Z}"] & H' \arrow[d, "m_{H' \to H}"] \\
  G  & Z \arrow[l, "m_{Z \to G}" below] \arrow[r, "m_{Z \to H}" below] & H
\end{tikzcd}
\]
where $(3')$ is constructed as pullback, $(3'')$ is a pullback, because $(\widetilde{3})$ and $(3')$ are using Fact~\ref{fact:composition-lemma-pb} and $(6'')$ is constructed as pushout yielding a mediating morphism $m_{H' \to H} \colon H' \to H$ such that $(6')$ commutes and $m_{\widetilde{H} \to H} = m_{H' \to H} \circ m_{\widetilde{H} \to H'}$.
Finally,
$(6')$ is a pushout using Fact~\ref{fact:composition-lemma-po}, and both $(6')$ and $(6'')$ are pullbacks using Fact~\ref{fact:po-along-monos-are-pb}. 
  \end{enumerate}
\end{rem}

\begin{lem}
  \label{lemma:restriction-factors}\hfill
  \begin{enumerate}
    \item
    Consider the accessed parts \mbox{$\acc(d)\colon \ACC(d) \to G$} and \linebreak $\acc(\hat{d})\colon \ACC(\hat{d}) \to H$.
    Let $\ACC(d) \dder_p \widetilde{H}$ be the restriction of $G \dder_p H$ wrt $\acc(d)$.
    Then there is a monomorphism $m_{\ACC(\hat{d}) \to \widetilde{H}}\colon \ACC(\hat{d}) \to \widetilde{H}$ such that
    \[
    \acc(\hat{d}) = m_{\widetilde{H} \to H} \circ m_{\ACC(\hat{d}) \to \widetilde{H}}
    \]
    where $m_{\widetilde{H} \to H}$ is the monomorphism given by Lemma~\ref{lemma:restriction} in connection with Remark~\ref{remark:restriction}.

  \item
    \label{item:lemma-restriction-factors-two}
    Let $\acc(d)$ factor through $m_{G' \to G}$.
    Then there is a monomorphism \linebreak $m_{\ACC(\hat{d}) \to H'}\colon \ACC(\hat{d}) \to H'$ such that
    \[
    \acc(\hat{d}) = m_{H' \to H} \circ m_{\ACC(\hat{d}) \to H'}
    \]
    where $m_{H' \to H}$ is the monomorphism given by Lemma~\ref{lemma:restriction} in connection with Remark~\ref{remark:restriction}.
  \end{enumerate}
\end{lem}

\begin{proof}
  Statement 1.
  Consider the following pullbacks
  \[
\begin{tikzcd}[column sep=large]
  (Z \cap \ACC(\hat{d})) \cap (\ACC(\hat{d}) \cap h(R)) \arrow[r, "m_3"] \arrow[d, "m_4"] \arrow[dr, phantom, "(PB2)"]
  & \ACC(\hat{d}) \cap h(R) \arrow[ddr, phantom, "(PB1)"] \arrow[dr, "m_1"] \arrow[d, "m_2"] & \\
  Z \cap \ACC(\hat{d}) \arrow[drr, bend right = 10, "m_{Z \to H} \circ m_Z"] & \ACC(\hat{d}) \arrow[dr, "\acc(\hat{d})"]  & h(R) \arrow[d, "m_{h(R) \to H}"]\\
                            & & H
\end{tikzcd}
\]
where the underlying monomorphisms of $(PB1)$ are the accessed part $\acc(\hat{d})$ of $\hat{d}$ and the monomorphism of the epi-mono factorization of $h\colon R \to H$ and the underlying monomorphism of $(PB2)$ are $m_{h(R) \to H} \circ m_1$ and $m_{Z \to H} \circ m_Z$ given in Construction~\ref{construction:accessed-part} Step~\ref{item:construction-accessed-part-pb-one} and in $(PB1)$.
Let 
  \[
  \begin{tikzcd}
    (Z \cap \ACC(\hat{d})) \cap (\ACC(\hat{d}) \cap h(R))  \arrow[r, "m_3"] \arrow[d, "m_4"] \arrow[dr, phantom, "(PO)"]
& \ACC(\hat{d}) \cap h(R) \arrow[d, "m_6"] \\
\ACC(\hat{d}) \cap Z \arrow[r, "m_5"] 
& (Z \cap \ACC(\hat{d})) \cup (\ACC(\hat{d}) \cap h(R))
\end{tikzcd}
\]
be the pushout corresponding to $(PB2)$.
Using the distributivity of pullbacks of monomorphisms and the corresponding pushouts as well as the isomorphisms $H \iso Z \cup h(R)$, which is according to the Fact~\ref{fact:distributivity-lemma}, and $\ACC(\hat{d}) \cap H \iso \ACC(\hat{d})$, which holds trivially because $\ACC(\hat{d})$ is a subobject of $H$, one gets
\[(Z \cap \ACC(\hat{d})) \cup (\ACC(\hat{d}) \cap h(R)) \iso \ACC(\hat{d}) \cap (Z \cup h(R)) \iso \ACC(\hat{d}) \cap H \iso \ACC(\hat{d}).\]
Therefore, one can replace the pushout object in (PO) by $\ACC(\hat{d})$.

Next, one can show that the monomorphisms $m_{Z \to H} \circ m_Z$ and $m_{h(R) \to H} \circ m_1$ factor through $m_{\widetilde{H} \to H}$ given in Remark~\ref{remark:restriction}(\ref{item:remark-restriction-one}) as follows.

Ad $m_{Z \to H} \circ m_Z$:
Construction~\ref{construction:accessed-part} Step~\ref{item:construction-accessed-part-mediate} yields $m_{Z \to G} \circ m_Z = \acc(d) \circ \overline{m}_{\ACC(d) \cap Z}$, in particular.
Then the universal pullback property of ($\widetilde{3}$) in Remark~\ref{remark:restriction}(\ref{item:remark-restriction-one}) induces a mediating morphism
$m\colon \ACC(\hat{d}) \cap Z \to \widetilde{Z}$ such that $m_{\widetilde{Z} \to Z} \circ m = m_Z$, in particular.
This and the commutativity of ($\widetilde{6}$) yields
\[
m_{Z \to H} \circ m_Z = m_{Z \to H} \circ m_{\widetilde{Z} \to Z} \circ m = m_{\widetilde{H} \to H} \circ m_{\widetilde{Z} \to \widetilde{H}} \circ m.
\]

Ad $m_{h(R) \to H} \circ m_1$:
According to Remark~\ref{remark:restriction}(\ref{item:remark-restriction-one}), one has $m_{\widetilde{H} \to H} \circ \widetilde{h} = h$.
Let $m_{\widetilde{h}(R) \to \widetilde{H}} \circ e_{R \to \widetilde{h}(R)}$ be the epi-mono factorization of $\widetilde{h}$.
Then $m_{\widetilde{H} \to H} \circ m_{\widetilde{h}(R) \to \widetilde{H}} \circ e_{R \to \widetilde{h}(R)}$ is an epi-mono factorization of $h$.
As epi-mono factorizations are unique up to isomorphism, one can assume $m_{h(R) \to H} = m_{\widetilde{H} \to H} \circ m_{h(R) \to \widetilde{H}}$ without loss of generality.
This yields $m_{h(R) \to H} \circ m_1 = m_{\widetilde{H} \to H} \circ m_{h(R) \to \widetilde{H}} \circ m_1$.
As $m_{\widetilde{H} \to H}$ is a monomorphism, the commutativity of $(PB2)$ implies the commutativity of the following diagram
\[
\begin{tikzcd}
(Z \cap \ACC(\hat{d})) \cap (\ACC(\hat{d}) \cap h(R)) \arrow[r, "m_3"] \arrow[d, "m_4"]
& \ACC(\hat{d}) \cap h(R) \arrow[d, "m_{h(R) \to \widetilde{H}} \circ m_1"] \\
Z \cap \ACC(\hat{d}) \arrow[r, "m_{\widetilde{Z} \to \widetilde{H}} \circ m"] 
& \widetilde{H}
\end{tikzcd}
\]
Moreover, because $m_{\widetilde{H} \to H}$ is a monomorphism, this diagram is also a pullback.
The universal pushout property induces a mediating morphism $\acc(\hat{d})_{\widetilde{H}} \colon \ACC(\hat{d}) \to \widetilde{H}$ which is a monomorphism because it goes from the pushout to the codomain of the enclosing pullback (cf. Fact~\ref{fact:mediating-morphism-of-po-of-pb-is-mono}).

Statement 2.
According to Remark~\ref{remark:restriction}(\ref{item:remark-restriction-two}), there is a monomorphism $m_{\widetilde{H} \to H'}$ this can be composed with $\acc(\hat{d})_{\widetilde{H}}$ yielding $\acc(\hat{d})_{H'} = m_{\widetilde{H} \to H'} \circ \acc(\hat{d})_{\widetilde{H}}\colon \ACC(\hat{d}) \to H'$.
\end{proof}
  
\begin{construction}
  Let $\acc(d)$ factor through $m_{G' \to G}$.
  Then the \emph{restriction} of $d$ wrt $m_{G' \to G}$,
  denoted by $d|m_{G' \to G}$, can be constructed by induction on the length of $d$:

  \base
  $((G\dder^0_P X)|m_{G' \to G})= (G'\dder^0_P G')$, and

  \step\!$((G \dder^{n+1}_P X)|m_{G' \to G}) =\! ((G \dder_{p} H\dder^n_P X)|m_{G' \to G})=\!(G'\dder_{p} H' \dder^n_P X')$
  where $G' \dder_p H'$ is the single-step restriction of $G \dder_p H$ wrt $m_{G' \to G}$
  and $(H' \dder^n_P X') = (H \dder^n_P X)|m_{H' \to H}$ is given by the induction hypothesis for derivations of length $n$.
\end{construction}

\begin{rem}
  $H' \dder^n_P X'$ exists as $\acc(\hat{d})\colon \ACC(\hat{d}) \to H$ factors through $m_{H' \to H}$ according to Lemma~\ref{lemma:restriction-factors}(\ref{item:lemma-restriction-factors-two}).
\end{rem}

\begin{exa}\label{ex:restriction-accessed-part-color}
  The derivation $d_{\colr}$ can be restricted to its accessed part yielding the derivation
  \[
  d_{\colr}|\acc(d_{\colr}) = \mbox{
\newcommand{\ens}{15pt}

\begin{tikzpicture}[baseline=-1pt]
\node (g1) {
\begin{tikzpicture}
\node (n1) [state] {};
\node (n2) [state,right=\ens of n1] {};
\node (n3) [state,below=\ens of n1] {};
\node (n4) [state,right=\ens of n3] {};
\node (n5) [state,below=\ens of n3] {};
\node (n6) [state,right=\ens of n5] {};
\end{tikzpicture}
};
\node (g1g2) [right=30pt of g1] {$\dder^{6}_{\addloop}$};
\node (g2) [right=45pt of g1g2] {
\begin{tikzpicture}
\node (n1) [state] {}
	edge [loop left] node [left] {\scriptsize $a$} (n1);
\node (n2) [state,right=\ens of n1] {}
	edge [loop right] node [right] {\scriptsize $a$} (n2);
\node (n3) [state,below=\ens of n1] {}
	edge [loop left] node [left] {\scriptsize $a$} (n3);
\node (n4) [state,right=\ens of n3] {}
	edge [loop right] node [right] {\scriptsize $a$} (n4);
\node (n5) [state,below=\ens of n3] {}
	edge [loop left] node [left] {\scriptsize $a$} (n5);
\node (n6) [state,right=\ens of n5] {}
	edge [loop right] node [right] {\scriptsize $a$} (n6);
\end{tikzpicture}
};
\node (g2g3) [right=45pt of g2] {$\dder^{2}_{\addcolor}$};
\node (g3) [right=55pt of g2g3] {
\begin{tikzpicture}
\node (n1) [state] {}
	edge [loop left] node [left] {\scriptsize $a$} (n1);
\node (n2) [state,right=\ens of n1] {}
	edge [loop right] node [right] {\scriptsize $a$} (n2);
\node (n3) [state,below=\ens of n1] {}
	edge [loop left] node [left] {\scriptsize $a$} (n3);
\node (n4) [state,right=\ens of n3] {}
	edge [loop right] node [right] {\scriptsize $a$} (n4);
\node (n5) [state,below=\ens of n3] {}
	edge [loop left] node [left] {\scriptsize $a$} (n5);
\node (n6) [state,right=\ens of n5] {}
	edge [loop right] node [right] {\scriptsize $a$} (n6);
\node (n7) [state,left=25pt of n3] {}
	edge [loop above] node [above] {\scriptsize $1$} (n7);
\node (n8) [state,right=25pt of n4] {}
	edge [loop above] node [above] {\scriptsize $2$} (n8);
\end{tikzpicture}
};
\node (g3g4) [right=60pt of g3] {$\dder^{6}_{\choosecolr}$};
\node (g4) [right=50pt of g3g4] {
\begin{tikzpicture}
\node (n1) [state] {};
\node (n2) [state,right=\ens of n1] {};
\node (n3) [state,below=\ens of n1] {};
\node (n4) [state,right=\ens of n3] {};
\node (n5) [state,below=\ens of n3] {};
\node (n6) [state,right=\ens of n5] {};
\node (n7) [state,left=\ens of n3] {}
	edge [loop above] node [above] {\scriptsize $1$} (n7)
	edge [-] (n1)
	edge [-] (n3)
	edge [-] (n5);
\node (n8) [state,right=\ens of n4] {}
	edge [loop above] node [above] {\scriptsize $2$} (n8)
	edge [-] (n2)
	edge [-] (n4)
	edge [-] (n6);
\end{tikzpicture}
};
\end{tikzpicture}
}.
  \]
\end{exa}

\subsection{Composition of Restrictions}

In the next section, we need that the restriction of a restriction of a
given derivation is a restriction of the given derivation. This property
is shown in the following proposition. The proof is based on a further
property stating that a restriction preserves the accessed part.

\begin{prop}
  \label{prop:restriction-of-restriction-is-restriction}\hfill
  \begin{enumerate}
  \item
    Let $d=(G \dder^*_P X)$ be a derivation and $\acc(d)$ factor through
    $m_{G'\to G}\colon G'\to G$. Let $d|m_{G'\to G}$ be the associated restriction.
    Then there is an isomorphism $i\colon \ACC(d)\to \ACC(d|m_{G'\to G})$ such that
    $m_{G'\to G} \circ \acc(d|m_{G'\to G}) \circ i = \acc(d)$.
    
  \item
    Let $\acc(d|m_{G'\to G})$ factor through $m_{G''\to G'}$ and
    $(d|m_{G'\to G})|m_{G''\to G'}$ be the associated restriction.
    Then
    $(d|m_{G'\to G})|m_{G''\to G'} = d|m_{G'\to G} \circ m_{G''\to G'}$.
  \end{enumerate}
\end{prop}

\begin{proof}
  Both statements are proved by induction on the length of $d$.

  Statement 1.
  Base $n=0$:
  This means $d = (G \dder^0_P X), \acc(d) = \emptyset_G$ and $d|m_{G' \to G} = (G' \dder^0_P G')$.
  Therefore, $\acc(d|m_{G' \to G}) = \emptyset_{G'}$ and choosing $i = \id_\emptyset$, one gets
  $m_{G' \to G} \circ \emptyset_{G'} \circ \id_\emptyset = \emptyset_G$ as the initial morphism is unique.

  Step $n+1$:
  Consider $d = (G \dder^{n+1}_P X) = (G \dder_p H \dder^n_P X)$.
  According to Remark~\ref{remark:restriction}, the double pushout $(1)$ and $(2)$ can be decomposed into
\[
\begin{tikzcd}[column sep=3cm]
L \arrow[dr, phantom, "(\widetilde{4})" description] \arrow[ddd, controls={+(-3,0) and +(-3,0)}, "g" left] \arrow[d, "\widetilde{g}"] & K \arrow[dr, phantom, "(\widetilde{5})" description] \arrow[l, "l" below] \arrow[r, "r" below] \arrow[d, "\widetilde{z}"] & R \arrow[ddd, controls={+(3,0) and +(3,0)}, "h" right] \arrow[d, "\widetilde{h}"]\\
  \ACC(d) \arrow[dr, phantom, "(3'')"] \arrow[d, "m_{\ACC(d) \to G'}"] \arrow[dd, bend right, "\acc(d)" left] & \widetilde{Z} \arrow[dr, phantom, "(6'')"] \arrow[l, "m_{\widetilde{Z} \to \ACC(d)}" below] \arrow[r, "m_{\widetilde{Z} \to \widetilde{H}}" below] \arrow[d, "m_{\widetilde{Z} \to Z'}"] & \widetilde{H} \arrow[d, "m_{\widetilde{H} \to H'}"] \arrow[dd, controls={+(2,0) and +(2,0)}, "m_{\widetilde{H} \to H}" left] \\
    G' \arrow[dr, phantom, "(3')"] \arrow[d, "m_{G' \to G}" right] & Z' \arrow[dr, phantom, "(6')"] \arrow[l, "m_{Z' \to G'}" below] \arrow[r, "m_{Z' \to H'}" below] \arrow[d, "m_{Z' \to Z}"] & H' \arrow[d, "m_{H' \to H}"] \\
  G  & Z \arrow[l, "m_{Z \to G}" below] \arrow[r, "m_{Z \to H}" below] & H
\end{tikzcd}
\]
Let $d' = d|m_{G' \to G}$ and $\hat{d}' = \hat{d}|m_{H' \to H}$.
According to Lemma~\ref{lemma:restriction-factors}(\ref{item:lemma-restriction-factors-two}), there is a monomorphism $m_{\ACC(\hat{d}) \to H'}$ such that
$\acc(\hat{d}) = m_{H' \to H} \circ m_{\ACC(\hat{d}) \to H'}$ for $\hat{d} = (H \dder^n_P X)$.
Consequently, the induction hypothesis applies to $\hat{d}$ yielding an isomorphism $j \colon \ACC(\hat{d}) \to \ACC(\hat{d}')$ such that
\[
m_{H' \to H} \circ \acc(\hat{d}') \circ j = \acc(\hat{d}).
\]\newpage
\noindent The following diagram without the dashed arrows shows the given situation
\[
\begin{tikzcd}
     & & L \arrow[d, "g'"] \arrow[dl, bend right, "e_{g'}"] \arrow[dll, bend right, "e_g"] & K \arrow[d, "z'"] \arrow[l, "l" below] \arrow[rr, "r" below] & & R \arrow[d, "h'"] \\
g(L) \arrow[r, dashed, bend left, "i_L"] \arrow[dddrr, bend right, "m_g"] & g'(L) \arrow[l, dashed, bend left, "i_L^{-1}"] \arrow[r, "m_{g'}"] & G' \arrow[ddd, "m_{G' \to G}"] & Z' \arrow[ddd, "m_{Z' \to Z}"] \arrow[l, "m_{Z' \to G'}" below] \arrow[rr, "m_{Z' \to H'}" below] & & H' \arrow[dl, phantom, "(PB')"] \arrow[ddd, "m_{H' \to H}"] \\
&  &    &    & \ACC(\hat{d}') \cap Z' \arrow[rr, dashed, "m_3" near start] \arrow[ul, dashed, "m_4"] \arrow[d, dashed, bend left, "i_1"] & & \ACC(\hat{d}') \arrow[ul, "\acc(\hat{d}')" above right] \arrow[d, bend left, "j^{-1}"] \\
&  &    &    & \ACC(\hat{d}) \cap Z  \arrow[u, dashed, bend left, "i_3"] \arrow[uul, dashed, bend left = 10, "i_2"] \arrow[rr, dashed, "m_1" near start] \arrow[dl, dashed, "m_2"] & & \ACC(\hat{d}) \arrow[dl, "\acc(\hat{d})"] \arrow[u, bend left, "j"] \\
& & G  & Z \arrow[l, "m_{Z \to G}" below] \arrow[rr, "m_{Z \to H}" below]  & & H \arrow[ul, phantom, "(PB)"]
\end{tikzcd}
\]
Then one can construct the pullbacks $(PB)$ and $(PB')$.
Using the universal property of $(PB), (6')$ and $(PB')$ and the respective commutativities, $i_1, i_2$ and $i_3$ are induced as mediating morphisms:

\noindent
$i_1$:
Using the commutativity of $PB'$ and the induction hypothesis, the equation
\begin{align*}
  m_{Z \to H} \circ m_{Z' \to Z} \circ m_4
  & = m_{H' \to H} \circ m_{Z' \to H'} \circ m_4\\
  & = m_{H' \to H} \circ \acc(\hat{d}') \circ m_3\\ 
  & = \acc(\hat{d}) \circ j^{-1} \circ m_3
\end{align*}
holds. By applying the universal pullback property the above equation implies
\begin{align*}
m_1 \circ i_1 & = j^{-1} \circ m_3\\
m_2 \circ i_1 & = m_{Z' \to Z} \circ m_4.
\end{align*}

\noindent
$i_2$:
Using the commutativity of $PB$ and the induction hypothesis one gets
\begin{align*}
  m_{Z \to H} \circ m_2 & = \acc(\hat{d}) \circ m_1\\
  & = m_{H' \to H} \circ \acc(\hat{d}') \circ j \circ m_1
\end{align*}
By applying the universal pullback property the equation implies
\begin{align*}
m_{Z' \to H'} \circ i_2 &= \acc(\hat{d}') \circ j \circ m_1\\
m_{Z' \to Z} \circ i_2  & = m_2.
\end{align*}

\noindent
$i_3$:
By applying the universal pullback property the equation $\acc(\hat{d}') \circ j \circ m_1 = m_{Z' \to H'} \circ i_2$ implies
\begin{align*}
m_3 \circ i_3 & = j \circ m_1\\
m_4 \circ i_3 &= i_2.
\end{align*}

Considering these commutativities, one gets
\begin{align*}
  m_1 \circ i_1 \circ i_3 = m_1\\
  m_2 \circ i_1 \circ i_3 = m_2\\
  m_3 \circ i_3 \circ i_1 = m_3\\
  m_4 \circ i_3 \circ i_1 = m_4
\end{align*}
so that
$i_1 \circ i_3 = \id_{\ACC(\hat{d}) \cap Z}$ and
$i_3 \circ i_1 = \id_{\ACC(\hat{d}') \cap Z'}$
meaning that $i_1$ and $i_3$ are inverse to each other.
Together with the commutativity of $(3')$, one gets $(*)$
\begin{align*}
  m_{G' \to G} \circ m_{Z' \to G'} \circ m_4 \circ i_3 & = m_{Z \to G} \circ m_{Z' \to Z} \circ m_4 \circ i_3\\
  & = m_{Z \to G} \circ m_{Z' \to Z} \circ i_2 \\
  & = m_{Z \to G} \circ m_2.
\end{align*}

On the other hand, consider the epi-mono factorizations
$g' = m_{g'} \circ e_{g'}$ and $g = m_g \circ e_g$.
As $g = m_{G' \to G} \circ g'$ and $g$ has a unique epi-mono factorization, there is an isomorphism
$i_L \colon g(L) \to g'(L)$ such that $(**)$
\[
m_g = m_{G' \to G} \circ m_{g'} \circ i_L.
\]
Using the pushout property of
$\ACC(d) = g(L) \cup (\ACC(\hat{d}) \cap Z)$
together with $(*)$ and $(**)$, one gets an isomorphism
\[
i \colon \ACC(d) = g(L) \cup (\ACC(\hat{d}) \cap Z) \to g'(L) \cup (\ACC(\hat{d}') \cap Z') = \ACC(d') = \ACC(d|m_{G' \to G}).
\]
such that
$m_{G'\to G} \circ \acc(d|m_{G'\to G}) \circ i = \acc(d)$.
    
This proves the first part of the proposition.

Statement 2.
Base $n=0$:
The only derivation of length 0 starting in $G''$ is $G'' \dder^0_P G''$ such that the statement holds.

Step $n+1$:
Let
$d = (G \dder_p H \dder^n_P X)$, 
$d' = d|m_{G' \to G} = (G' \dder_p H' \dder^n_P X')$, and
$d'' = d'|m_{G'' \to G'} = (G'' \dder_p H'' \dder^n_P X'')$, as well as
$\hat{d} = (H \dder^n_P X)$,
$\hat{d}' = (H' \dder^n_P X')$,
and
$\hat{d}'' = (H'' \dder^n_P X'')$.
Consider the decomposition of the double pushout defining $G \dder_p H$ wrt $m_{G' \to G}$ as given in Lemma~\ref{lemma:restriction} and the respective decomposition of the double pushout defining $G' \dder_p H'$ wrt $m_{G'' \to G'}$
\[
\begin{tikzcd}[column sep=3cm]
  L \arrow[dr, phantom, "(4'')" description] \arrow[dd, bend right, "g'" left] \arrow[d, "g''"] & K \arrow[dr, phantom, "(5'')" description] \arrow[l, "l" below] \arrow[r, "r" below] \arrow[d, "z''"] & R \arrow[dd, bend left, "h'" right] \arrow[d, "h''" left]\\
  G'' \arrow[dr, phantom, "(3'')" description] \arrow[d, "m_{G'' \to G'}" right] & Z'' \arrow[dr, phantom, "(6'')" description] \arrow[l, "m_{Z'' \to G''}" below] \arrow[r, "m_{Z'' \to H''}" below] \arrow[d, "m_{Z'' \to Z'}"] & H'' \arrow[d, "m_{H'' \to H'}" left] \\
  G'  & Z' \arrow[l, "m_{Z' \to G'}" below] \arrow[r, "m_{Z' \to H'}" below] & H'\\
\end{tikzcd}
\]
The composition property of pullbacks yields that $(3'') + (3)$ and $(6'') + (6)$ are pullbacks that coincide thereby with the double pushout defining $G'\dder_p H'$.
This shows that $G'' \dder_p H''$ is the single-step restriction of $G \dder_p H$ wrt $m_{G' \to G} \circ m_{G'' \to G'}$.

According to the first point of the proposition,
$\ACC(d)$ and $\ACC(d')$ can be considered as equal without loss of generality with $m_{G' \to G} \circ \acc(d') = \acc(d)$.
By assumption, $\acc(d') = \acc(d|m_{G' \to G})$ factors through $m_{G'' \to G'}$ such that
\[
m_{G' \to G} \circ \acc(d') = m_{G' \to G} \circ m_{G'' \to G'} \circ \acc(d'') = \acc(d)
\]
meaning that the restriction $d|m_{G' \to G} \circ m_{G'' \to G'}$ is defined.
Moreover, it is equal to
$(d|m_{G' \to G})|m_{G'' \to G'}$, as this holds for the first steps shown above and for the tails $\hat{d},\hat{d}', \hat{d}''$ by induction hypothesis.
This completes the proof of statement 2.
\end{proof}

\subsection{Spines}

As the restriction of a derivation to its accessed part plays a prominent role in the next section, we name it.

\begin{defi}
The restriction $d|\acc(d)$ of a derivation $d$ to its accessed part is called \emph{spine} of $d$ and denoted by $\spine(d)$.
\end{defi}

\begin{rem}
Using the definitions of accessed parts and restrictions, the spine can be characterized as follows:
  
\base
  $\spine(G\dder^0 G)=(\emptyset \dder^0 \emptyset)$, and

\step
  $\spine(G\dder_{p} H\dder^n_P X)= (\ACC(G \dder_{p} H \dder^n_P X) \dder_{p} \widetilde{H} \dder^n_P \widetilde{X})$
where $p$ is applied to $\ACC(G \dder_{p} H \dder^n_P X)$ using the
morphism $\widetilde{g}=\overline{m}_{g(L)} \circ e_{L \to g(L)}$ (given in Remark~\ref{remark:restriction}(\ref{item:remark-restriction-one})) as left matching morphism and
$\widetilde{H} \dder^n_P \widetilde{X}$ is the restriction $(H \dder^n_P X)|m_{\widetilde{H} \to H}$.
\end{rem}

\begin{exa}\label{ex:spine-color-dual}
As the derivation $d_{\colr}$ is restricted to its accessed part in Example \ref{ex:restriction-accessed-part-color}, the derivation given there is the spine of $d_{\colr}$.
Moreover, it is the spine of $\move(d_{\colr},d_{\dual})$ given in Example \ref{ex:moved-variant-color-dual}.
This is not a coincidence.
In the next section, we show that the spine of a derivation equals the spine of each of its moved variants.
\end{exa}

\section{Moving Preserves the Spine}
\label{sec:moving-preserves-the-spine}
In this section, we relate the moving of a
derivation along a derivation
and the restriction of a derivation by showing that
moving preserves the
spine of the moved derivation. As all constructions
that are involved in
derivations are only unique up to isomorphism, the
spine of a
derivation and of its moved variant can only be
equal up to isomorphism.
To define the equality, we employ the special case of restrictions with
isomorphic restriction morphism. Consequently, two derivations are equal up
to isomorphism if they apply the same rules in the same order and all
corresponding derived and intermediate objects are isomorphic in such a way
that the isomorphisms are compatible with the matching morphisms and with the
embeddings of the intermediate morphisms into the derived objects.
  
\begin{defi}
  Let
  $d = (G\dder^*_P X)$ and $d' = (G' \dder^*_P X')$ be two derivations
  and $iso:G' \to G$ be an isomorphism such that $d=d'|iso$.
  Then $d$ and $d'$ are called
  \emph{equal up to isomorphism}, denoted by $d \equiv d'$.
\end{defi}

\begin{rem}    
The accessed part $\acc(d)\colon \ACC(d)\to G$ factorizes always through an
isomorphic restriction morphism: $\acc(d) = id_G \circ \acc(d) = iso \circ iso^{-1} \circ \acc(d)$
so that $d|iso$ is defined.
According to the construction of restrictions, each
derivation step decomposes as given in Lemma~\ref{lemma:restriction} for a single step, i.e., if
$d = (G_0 \dder_{r_1} G_1 \dder_{r_2} \ldots \dder_{r_n} G_n)$,
then the $i$-th step for $i=1,\ldots,n$ has the form
  \[
\begin{tikzcd}[column sep=3cm]
  L_i \arrow[dr, phantom, "(4_i)" description] \arrow[dd, controls={+(-3,0) and +(-3,0)}, "g_i" left] \arrow[d, "g_i'"] & K_i \arrow[dr, phantom, "(5_i)" description] \arrow[l, "l_i" below] \arrow[r, "r_i" below] \arrow[d, "z_i'"] & R_i \arrow[dd, controls={+(+3,0) and +(+3,0)}, "h_i" right] \arrow[d, "h_i'"]\\
  G'_{i-1} \arrow[dr, phantom, "(3_i)" description] \arrow[d, "m_{G'_{i-1} \to G_{i-1}}" left] & Z_i' \arrow[dr, phantom, "(6_i)" description] \arrow[l, "m_{Z_i' \to G_i'}" below] \arrow[r, "m_{Z_i' \to H_i'}" below] \arrow[d, "m_{Z_i' \to Z_i}" near end] & G_i' \arrow[d, "m_{G_i' \to G_i}"] \\
  G_{i-1}  & Z_i \arrow[l, "m_{Z_i \to G_{i-1}}" below] \arrow[r, "m_{Z \to H}" below] & G_i
\end{tikzcd}
\]
The diagrams $(4_i)$ and $(5_i)$ form the double pushouts of the derivation $d'$ if
the restriction morphism $m_{G_0'\to G_0}$ is chosen as the given isomorphism $iso$
with $G_0=G$ and $G_0'=G'$. Moreover, as the diagrams $(3_i)$ and $(6_i)$ are
pullbacks and pushouts, all the vertical morphisms are isomorphisms and the
commutativity of the diagrams means that the isomorphisms are compatible with
the matching morphisms and the horizontal embedding morphism. Therefore, there
is, in particular, a 1-to-1 correspondence between the sequentially
independent steps of $d$ and those of $d'$. Obviously, as the class of
isomorphisms is closed under identities, inversion, and sequential
composition, the introduced equality of derivations up to isomorphism is an equivalence relation.
\end{rem}

\begin{lem}
  \label{lemma:moving-preserves-spine-one-step}
  Let $(P,\overline{P})$ be a parallel independent pair of sets of rules.
  Let $d = (G \dder^*_P X)$ be a derivation and $\overline{d} = (G \dder_{\overline{p}} G')$ be a direct derivation for $\overline{p} \in \overline{P}$ with the intermediate object $\overline{Z}$ and its monomorphisms $m_{\overline{Z} \to G} \colon \overline{Z} \to G$ and $m_{\overline{Z} \to G'} \colon \overline{Z} \to G'$.
  Let $\acc(d) \colon \ACC(d) \to G$ and $\acc(d')\colon \ACC(d') \to G'$ be the accessed parts of $d$ and $d' = \move(d,\overline{d})$, respectively.
  Then
  \begin{enumerate}
  \item
    \label{item:moving-preserves-spine-one-step-item-one}
    there is a derivation $d''=(\overline{Z} \dder^*_P Z'')$ for some object $Z''$
    such that
    \[
    d|m_{\overline{Z} \to G} = d'' = d'|m_{\overline{Z} \to G'}.
    \]
  \item
    \label{item:moving-preserves-spine-one-step-item-two}
    Moreover, there exists an isomorphism
    $i\colon \ACC(d) \to \ACC(d')$
    and two monomorphisms $m_{\ACC(d) \to \overline{Z}} \colon \ACC(d) \to \overline{Z}$
    and $m_{\ACC(d') \to \overline{Z}} \colon \ACC(d') \to \overline{Z}$ such that
    \begin{align*}
      m_{\ACC(d') \to \overline{Z}} \circ i &= m_{\ACC(d) \to \overline{Z}}\\
      m_{\overline{Z} \to G} \circ m_{\ACC(d) \to \overline{Z}} &= \acc(d)\\
      m_{\overline{Z} \to G'} \circ m_{\ACC(d') \to \overline{Z}} &= \acc(d').
    \end{align*}
  \end{enumerate}
  Moreover, $\spine(d) \equiv \spine(d')$.
\end{lem}

\begin{proof}
  Statement 1.
  The statement is proved by induction on the length $n$ of the derivation $d$.

  \base
    $n=0$:
    This means that $d = (G \dder^0_P X)$, $d' = (G' \dder^0_P X')$,
    $\ACC(d) = \ACC(d') = \emptyset$,
    $\acc(d)=\emptyset_G$, $\acc(d')=\emptyset_{G'}$ and
    $\spine(d)= (\emptyset \dder^0_P \emptyset)=spine (d')$.
    Then, obviously, $\overline{Z} \dder^0_P \overline{Z}$ is a common
    restriction of $d$ and $d'$ wrt $m_{\overline{Z} \to G}$ and $m_{\overline{Z} \to G'}$ respectively.
    Moreover,
    choosing $i = \id_{\emptyset}$ and $m_{\emptyset \to \overline{Z}} = \emptyset_{\overline{Z}}$, one gets obviously:
    $m_{\emptyset \to \overline{Z}} \circ i = \emptyset_{\overline{Z}} \circ \id_{\emptyset} = \emptyset_{\overline{Z}}$.
    
  \step $n+1$:
    The derivation $d$ can be decomposed into the first direct derivation $G \dder_p H$ and the remaining section $\hat{d} = \tail(d) = (H \dder^n_P X)$.
    Let
    \[
    \begin{tikzcd}[column sep=huge]
      L \arrow[d, "g" left] \arrow[dr, phantom, "(1)" description] & K \arrow[dr, phantom, "(2)" description] \arrow[d, "z" left] \arrow[l, "l" below] \arrow[r, "r" below] & R \arrow[d, "h" left] &
      \overline{L} \arrow[d, "\overline{g}" left] \arrow[dr, phantom, "(\overline{1})" description] & \overline{K} \arrow[dr, phantom, "(\overline{2})" description] \arrow[d, "\overline{z}" left] \arrow[l, "\overline{l}" below] \arrow[r, "\overline{r}" below] & \overline{R} \arrow[d, "\overline{h}" left] \\
      G & Z \arrow[l, "m_{Z \to G}" below] \arrow[r, "m_{Z \to H}" below] & H
      & G & \overline{Z} \arrow[l, "m_{\overline{Z} \to G}" below] \arrow[r, "m_{\overline{Z} \to G'}" below] & G'
    \end{tikzcd}
    \]
be the double pushouts defining the two given direct derivations.
As they are parallel independent,
there are monomorphisms
$f\colon L \to \overline{Z}$ and $\overline{f}\colon \overline{L} \to Z$ such that
\begin{align*}
& (3)~ m_{\overline{Z} \to G} \circ f = g\\
& (\overline{3})~ m_{Z \to G} \circ \overline{f} = \overline{g}.
\end{align*}
Using Lemma~\ref{lemma:restriction}, the direct derivation $G \dder_p H$ can be restricted wrt
$m_{\overline{Z} \to G}$ and the direct derivation $G \dder_{\overline{p}} G'$ can be restricted wrt $m_{Z \to G}$ given by the two double pushouts
\[
\begin{tikzcd}[column sep=huge]
  L \arrow[d, "f" left] \arrow[dr, phantom, "(4)" description] & K \arrow[dr, phantom, "(5)" description] \arrow[d, "f''" left] \arrow[l, "l" below] \arrow[r, "r" below] & R \arrow[d, "f'" left]
  & \overline{L} \arrow[d, "\overline{f}" left] \arrow[dr, phantom, "(\overline{4})" description] & \overline{K} \arrow[dr, phantom, "(\overline{5})" description] \arrow[d, "\overline{f}''" left] \arrow[l, "\overline{l}" below] \arrow[r, "\overline{r}" below] & \overline{R} \arrow[d, "\overline{f}'" left] \\
  \overline{Z} & \hat{Z} \arrow[l, "m_{\hat{Z} \to \overline{Z}}" below] \arrow[r, "m_{\hat{Z} \to \overline{Z}'}" below] & \overline{Z}'
  & Z & \hat{\hat{Z}} \arrow[l, "m_{\hat{\hat{Z}} \to Z}" below] \arrow[r, "m_{\hat{\hat{Z}} \to Z'}" below] & Z'
\end{tikzcd}
\]
Moreover, there are the pullbacks
\[
\begin{tikzcd}[column sep=huge]
  \overline{Z} \arrow[d, "m_{\overline{Z}\to G}" left] \arrow[dr, phantom, "(6)" description] & \hat{Z} \arrow[dr, phantom, "(7)" description] \arrow[d, "m_{\hat{Z} \to Z}" left] \arrow[l, "m_{\hat{Z} \to \overline{Z}}" below] \arrow[r, "m_{\hat{Z} \to \overline{Z}'}" below] & \overline{Z}' \arrow[d, "m_{\overline{Z}' \to H}" right]
  & Z \arrow[d, "m_{Z\to G}" left] \arrow[dr, phantom, "(\overline{6})" description] & \hat{\hat{Z}} \arrow[dr, phantom, "(\overline{7})" description] \arrow[d, "m_{\hat{\hat{Z}} \to \overline{Z}}" left] \arrow[l, "m_{\hat{\hat{Z}} \to Z}" below] \arrow[r, "m_{\hat{\hat{Z}} \to Z'}" below] & Z' \arrow[d, "m_{Z' \to G'}" right] \\
  G & Z \arrow[l, "m_{Z \to G}" below] \arrow[r, "m_{Z \to H}" below] & H
  & G & \overline{Z} \arrow[l, "m_{\overline{Z} \to G}" below] \arrow[r, "m_{\overline{Z} \to G'}" below] & G'
\end{tikzcd}
\]
such that $(4) + (6) = (1), (5) + (7) = (2), (\overline{4}) + (\overline{6}) = (\overline{1})$ and $(\overline{5}) + (\overline{7}) = (\overline{2})$.
In particular, $(6)$ and $(\overline{6})$ coincide meaning that $\hat{Z} = \hat{\hat{Z}}$ without loss of generality.
Adding the pushout
\[
\begin{tikzcd}
  \hat{Z} \arrow[d, "m_{\hat{Z} \to \overline{Z}}" left] \arrow[r, "m_{\hat{Z} \to Z'}" below] \arrow[dr, phantom, "(8)" description] & Z' \arrow[d, "m_{Z' \to H'}" right] \\
  \overline{Z}' \arrow[r, "m_{\overline{Z}' \to H'}" below]  & H'
\end{tikzcd}
\]
which is also a pullback because it consists of four monomorphisms, one can regroup some of the diagrams
\vspace*{10pt}
\[
\begin{tikzcd}[row sep=large, column sep=huge]
  L \arrow[d, "f" left] \arrow[dr, phantom, "(4)" description] & K \arrow[dr, phantom, "(5)" description] \arrow[d, "f''" left] \arrow[l, "l" below] \arrow[r, "r" below] & R \arrow[d, "f'" left]
  & \overline{L} \arrow[d, "\overline{f}" left] \arrow[dr, phantom, "(\overline{4})" description] & \overline{K} \arrow[dr, phantom, "(\overline{5})" description] \arrow[d, "\overline{f}''" left] \arrow[l, "\overline{l}" below] \arrow[r, "\overline{r}" below] & \overline{R} \arrow[d, "\overline{f}'" left] \\
  \overline{Z} \arrow[d, "m_{\overline{Z} \to G'}" very near end] \arrow[dr, phantom, "(\overline{7})" description] & \hat{Z} \arrow[d, "m_{\hat{Z} \to Z'}" very near end] \arrow[l, "m_{\hat{Z} \to \overline{Z}}" below] \arrow[r, "m_{\hat{Z} \to \overline{Z}'}" below] \arrow[dr, phantom, "(8)" description] & \overline{Z}' \arrow[d, "m_{\overline{Z}' \to H'}" very near end]
  & Z \arrow[d, "m_{Z \to H}" very near end] \arrow[dr, phantom, "(7)" description] & \hat{\hat{Z}} \arrow[d, "m_{\hat{\hat{Z}} \to \overline{Z}'}" very near end] \arrow[l, "m_{\hat{\hat{Z}} \to Z}" below] \arrow[r, "m_{\hat{\hat{Z}} \to Z'}" below] \arrow[dr, phantom, "(8)" description] & Z' \arrow[d, "m_{Z' \to H'}" very near end] \\
  G' & Z' \arrow[l, "m_{Z' \to G'}" below] \arrow[r, "m_{Z' \to H}" below] & H'
  & H & \overline{Z}' \arrow[l, "m_{\overline{Z}' \to H}" below] \arrow[r, "m_{\overline{Z}' \to H'}" below] & H'
\end{tikzcd}
\]
\vspace*{10pt}

\noindent According to the composition properties of pushouts with pullbacks
$(4) + (\overline{7})$, $(5) + (8)$,
$(\overline{4}) + (7)$ and $(\overline{5}) + (8)$ are pushouts defining the direct derivations
$G' \dder_{p} H'$
and
$H \dder_{\overline{p}} H'$.
Together with the given direct derivations
$G \dder_{p} H$ and $G \dder_{\overline{p}} G'$,
they form the local Church-Rosser property of parallel independent direct derivations.
Moreover, the pushouts $(4)$ and $(5)$ are components of the pushouts $(4) + (7)$ and $(5) + (8)$ such that the direct derivation $\overline{Z} \dder_p \overline{Z}'$ given by $(4)$ and $(5)$ is a common (single-step) restriction of $G \dder_p H$ and $G' \dder_p H'$.
Using the induction hypothesis for $\hat{d} = \tail(d) = (H \dder^n_P X)$
  and $\overline{d}' = (H \dder_{\overline{p}} H')$, one gets, in particular,
  a derivation $\hat{d}'' = (\overline{Z}' \dder^n_P Z'')$
  such that
  $\hat{d}|m_{\overline{Z}' \to H} = \hat{d}'' = \hat{d}'|m_{\overline{Z}' \to H'}$
  with $\hat{d}' = \tail(d')$.
  The sequential composition $\overline{Z} \dder_p \overline{Z}' \dder^n_P Z''$
  is a common restriction of $d$ and $d'$ wrt
  $m_{\overline{Z} \to G}$ and $m_{\overline{Z} \to G'}$ respectively
  proving the first part of the lemma.

  Statement 2.
  It remains to show that there are an isomorphism $i\colon \ACC(d) \to \ACC(d')$
  and monomorphisms $m_{\ACC(d) \to \overline{Z}'} \colon \ACC(d) \to \overline{Z}'$
  and $m_{\ACC(d') \to \overline{Z}} \colon \ACC(d') \to \overline{Z}$
  with $m_{\ACC(d') \to \overline{Z}} \circ i = m_{\ACC(d) \to \overline{Z}}$.

  By induction hypothesis, there are an isomorphism
  $j\colon \ACC(\hat{d}) \to \ACC(\hat{d}')$
  and monomorphisms
  $m_{\ACC(\hat{d}) \to \overline{Z}'} \colon \ACC(\hat{d}) \to \overline{Z}'$
  and $m_{\ACC(\hat{d}') \to \overline{Z}'} \colon \ACC(\hat{d}') \to \overline{Z}'$
  such that
  \vspace*{10pt}
  \begin{align*}
    m_{\ACC(\hat{d}') \to \overline{Z}'} \circ j &= m_{\ACC(\hat{d}) \to \overline{Z}'}\\
    m_{\overline{Z}' \to H} \circ m_{\ACC(\hat{d}) \to \overline{Z}'} &= \acc(\hat{d})\\
    m_{\overline{Z}' \to H'} \circ m_{\ACC(d') \to \overline{Z}'} &= \acc(\hat{d}').
  \end{align*}\newpage
  \noindent Consider the pullbacks $(9)$ and $(10)$ put into the context of diagrams given above where the two squares in the background are the pullbacks $(7)$ and $(8)$
  \[
  \begin{tikzcd}[column sep=huge,row sep=large]
    Z \arrow[dd, "m_{Z \to H}" right] \arrow[dddr, phantom, "(9)"] &                     & \hat{Z} \arrow[ll, "m_{\hat{Z} \to Z}" below] \arrow[rr, "m_{\hat{Z} \to Z'}" below] \arrow[dd, "m_{\hat{Z} \to \overline{Z}'}" right]      &                       & Z' \arrow[dd, "m_{Z' \to H'}" left] \arrow[dddl, phantom, "(10)"] \\
      & \ACC(\hat{d}) \cap Z \arrow[ul, "m_1"] \arrow[ur, dashed, "m_5"] \arrow[dd, "m_2" near start] \arrow[rr, crossing over, bend left = 10, "\hat{i}^{-1}" near start] &                & \ACC(\hat{d}') \cap Z' \arrow[ul, dashed, "m_6" above right] \arrow[ur, "m_3"] \arrow[dd,  "m_4" near start] \arrow[ll, crossing over,  bend left = 10, "\hat{i}" near start] &    \\
    H &                     & \overline{Z}' \arrow[ll, "m_{\overline{Z}' \to H}" near end] \arrow[rr, "m_{\overline{Z} \to H'}" near end]  &                    & H' \\
    & \ACC(\hat{d}) \arrow[ul, "\acc(\hat{d})" below left] \arrow[from=uu, crossing over] \arrow[ur, "m_{\ACC(\hat{d}) \to \overline{Z}'}"] \arrow[rr, bend left = 10, "j^{-1}" above]       &                & \ACC(\hat{d}') \arrow[ul, "m_{\ACC(\hat{d}') \to \overline{Z}'}" above right] \arrow[from=uu, crossing over] \arrow[ur, "\acc(\hat{d}')" below right] \arrow[ll, bend left = 10, "j" below]    &    
  \end{tikzcd}
  \]
  The universal property of the pullbacks $(7), (8), (9)$ and $(10)$ induces the monomorphisms $m_5, m_6, \hat{i}$ and $\hat{i}^{-1}$ as mediating morphisms making the four triangles with $m_5$ and $m_6$ commutative in particular.
  The pullback property of $(9)$ and $(10)$ implies also that $\hat{i} \circ \hat{i}^{-1}$ and $\hat{i}^{-1} \circ \hat{i}$ are identities and, therefore,
  $\ACC(\hat{d}) \cap Z \iso \ACC(\hat{d}') \cap Z'$.
  Composing $m_5$ and $m_6$ with $m_{\hat{Z} \to \overline{Z}}$, one gets the commutative diagram
  \[
  \begin{tikzcd}
    & \overline{Z} & \\
    \ACC(\hat{d}) \cap Z \arrow[ur, "m_{\hat{Z} \to \overline{Z}} \circ m_5"] \arrow[rr, "\hat{i}" below] && \ACC(\hat{d}') \cap Z' \arrow[ul, "m_{\hat{Z} \to \overline{Z}} \circ m_6" above right]
  \end{tikzcd}
  \]
  Consider now the left matching morphisms
  $g$ of $G \dder_p H$ and
  $g'$ of $G' \dder_{\overline{p}} H'$.
  According to the equation $(3)$ and the definition of $g'$, one has
  $g = m_{\overline{Z} \to G} \circ f$ and
  $g' = m_{\overline{Z} \to G'} \circ f$.
  This implies
  \[
  g' = m_{\overline{Z} \to G'} \circ m_{f(L) \to \overline{Z}} \circ e_{f(L) \to \overline{Z}}
  \]
  using the epi-mono factorization $f = m_{f(L) \to \overline{Z}} \circ e_{f(L) \to \overline{Z}}$.
  As epi-mono factorizations are unique up to isomorphism, one gets the
  commutative diagram
  \[
  \begin{tikzcd}
    g(L) \arrow[dr, "m_{g(L) \to \overline{Z}}" below left] \arrow[r, phantom, "\iso" description] & f(L) \arrow[r, phantom, "\iso" description] & g'(L) \arrow[dl, "m_{g'(L) \to \overline{Z}}" below right] \\
    & \overline{Z} & 
  \end{tikzcd}
  \]
  Let $\widetilde{i}\colon g(L) \to g'(L)$ be the isomorphism and $\widetilde{i}^{-1}$ its inverse.
  Further, let $m_7 = m_{\hat{Z} \to \overline{Z}} \circ m_5$ and $m_8 = m_{\hat{Z} \to \overline{Z}} \circ m_6$.
  Then, after constructing the pullbacks of $m_{g(L) \to \overline{Z}}$ and $m_7$ and $m_{g'(L) \to \overline{Z}}$ and $m_8$, it follows, using the commutativities of the triangles and the universal property of pullbacks, that there exists an isomorphism $\widetilde{\widetilde{i}} \colon g(L) \cap \ACC(\hat{d}) \cap Z \to g'(L) \cap \ACC(\hat{d}') \cap Z'$.
  The situation is depicted in the following diagram.
  
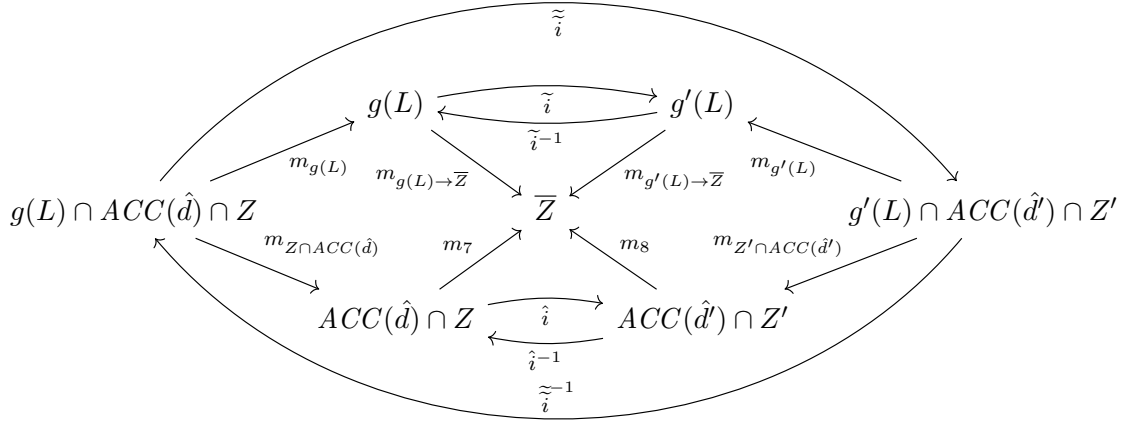
\[
  \begin{tikzcd}[column sep=small]
    & g(L) \arrow[dr, "m_{g(L) \to \overline{Z}}" below left] \arrow[rr, bend left = 10, "\widetilde{i}" below] & & g'(L) \arrow[dl, "m_{g'(L) \to \overline{Z}}" below right]  \arrow[ll, bend left = 10, "\widetilde{i}^{-1}"] \\
    g(L) \cap \ACC(\hat{d}) \cap Z  \arrow[rrrr, bend left = 50, "\widetilde{\widetilde{i}}" below] \arrow[dr, "m_{Z \cap \ACC(\hat{d})}" above right] \arrow[ur, "m_{g(L)}" below right] & & \overline{Z} &  & g'(L) \cap \ACC(\hat{d}') \cap Z' \arrow[dl, "m_{Z' \cap \ACC(\hat{d}')}" above left] \arrow[ul, "m_{g'(L)}" below left] \arrow[llll, bend left = 50, "\widetilde{\widetilde{i}}^{-1}" above]\\
    & \ACC(\hat{d}) \cap Z \arrow[ur, "m_7"]  \arrow[rr, bend left = 10, "\hat{i}" below] && \ACC(\hat{d}') \cap Z' \arrow[ul, "m_8" above right]  \arrow[ll, bend left = 10, "\hat{i}^{-1}"]
  \end{tikzcd}
  \]
  Constructing the pushouts $g(L) \cup (\ACC(\hat{d}) \cap Z)$ and $g'(L) \cup (\ACC(\hat{d}') \cap Z')$
  and using the universal property of pushouts, one gets the commutative diagram
  \[
  \begin{tikzcd}
    & \overline{Z} & \\
    \ACC(d) = g(L) \cup (\ACC(\hat{d}) \cap Z) \arrow[ur, "\acc(d)"] \arrow[rr, bend left = 10, "i" below] & & g'(L) \cup (\ACC(\hat{d}) \cap Z) = \ACC(d') \arrow[ul, "\acc(d')" above right] \arrow[ll, bend left = 10, "i^{-1}"]
  \end{tikzcd}
  \]
  where $i$ is an isomorphism and $i^{-1}$ its inverse.
  This proves the second part of the lemma.
\end{proof}

Based on this lemma we can formulate our main theorem.

\begin{thm}\label{thm:moving-preserves-spine}
Let $(P,\overline{P})$ be a parallel independent pair of sets of rules.
Let $d= (G \dder^*_P X)$ and $\overline{d} = (G\dder^*_{\overline{P}}G')$ be two derivations.
Then
\[
\spine(d) \equiv \spine(\move(d,\overline{d})).
\]
\end{thm}

\begin{proof}
The statement is proved by induction on the length $m$ of the
derivations $G \dder^*_{\overline{P}} G'$ for an arbitrary derivation $d$.

Base $m = 0$:
Then $\move(d,\overline{d}) = d$ by construction.
Therefore,
\[
\spine(d)=\spine(\move(d, \overline{d})).
\]

Step $m+ 1$:
Then $\overline{d}$ can be decomposed into
$\head(\overline{d}) = (G \dder_{\overline{p}} \overline{H})$ and
$\tail(\overline{d}) = (\overline{H} \dder^m_{\overline{P}} G')$.
By construction of moving, one gets the two derivations $\move(d, \head(\overline{d}))$ and
$\move(\move(d, \head(\overline{d})), \tail(\overline{d})) = \move(d, \overline{d})$.
As $\tail(\overline{d})$ has length $m$, the induction hypothesis is applicable such that one gets
\[
\spine(\move(\move(d, \head(\overline{d})), \tail(\overline{d}))) \equiv \spine(\move(d,\head(\overline{d}))).
\]
Let $d'=\move(d,\head(\overline{d}))$.
Using Lemma~\ref{lemma:moving-preserves-spine-one-step}(\ref{item:moving-preserves-spine-one-step-item-one}), one gets
$d''=(\overline{Z} \dder^*_P Z'')$ for some
object $Z''$ such that
\[
d|m_{\overline{Z} \to G} = d'' = d'|m_{\overline{Z} \to G'}
\]
where
$\overline{Z}$ is the intermediate object of $\head(\overline{d})$ and
$m_{\overline{Z} \to G}\colon \overline{Z} \to G$ and
$m_{\overline{Z} \to \overline{H}}\colon \overline{Z} \to \overline{H}$ are the associated monomorphisms.
Using Lemma~\ref{lemma:moving-preserves-spine-one-step}(\ref{item:moving-preserves-spine-one-step-item-two}), one gets two monomorphisms
$m_{\ACC(d) \to \overline{Z}} \colon \ACC(d) \to \overline{Z}$ and
$m_{\ACC(d') \to \overline{Z}} \colon \ACC(d') \to \overline{Z}$ such that
\begin{align*}
  m_{\overline{Z}\to G} \circ m_{\ACC(d) \to \overline{Z}} &= \acc(d) \\
  m_{\overline{Z} \to G'} \circ m_{\ACC(d') \to \overline{Z}} &= \acc(d').
\end{align*}
Therefore, the restrictions
$d''|m_{\ACC(d) \to \overline{Z}}$ and
$d''|m_{\ACC(d') \to \overline{Z}}$ are defined.
Moreover, there is an isomorphism $i \colon \ACC(d) \to \ACC(d')$ such that
\[
m_{\ACC(d') \to \overline{Z}} \circ i = m_{\ACC(d) \to \overline{Z}}.
\]
This implies
\[
d''|m_{\ACC(d) \to \overline{Z}} \equiv d''|m_{\ACC(d') \to \overline{Z}}
\]
as restrictions are uniquely determined (up to isomorphism) by the underlying monomorphism. Using Proposition~\ref{prop:restriction-of-restriction-is-restriction} and the equalities obtained above, one gets
\begin{align*}
\spine(d)
& = d|\acc(d) \\
& = d|m_{\overline{Z}\to G} \circ m_{\ACC(d) \to \overline{Z}} \\
& = (d|m_{\overline{Z} \to G})|m_{\ACC(d) \to \overline{Z}} \\
& = d''|m_{\ACC(d) \to \overline{Z}} \\
& \equiv d''|m_{\ACC(d') \to \overline{Z}} \\
& = (d'|m_{\overline{Z} \to G'})|m_{\ACC(d') \to \overline{Z}} \\
& = d'|m_{\overline{Z} \to G'} \circ m_{\ACC(d') \to \overline{Z}} \\
& = d'|\acc(d') \\
& = \spine(d') \\
& \equiv \spine(\move(d, \head(\overline{d}))) \\
& \equiv \spine(\move(\move(d, \head(\overline{d})), \tail(\overline{d})))\\
& \equiv \spine(\move(d,\overline{d})).
\end{align*}
This completes the proof.
\end{proof}

\begin{cor}\label{crl:evom-preserves-spine}
Let $(P,\overline{P})$ be a sequentially independent pair of sets of rules.
Let $\overline{d} = (G \dder^*_{\overline{P}} G')$ and $d' = (G' \dder^*_P X)$ be two derivations.
Then 
\[
\spine(d') \equiv \spine(\evom(d',\overline{d}))
\]
\end{cor}

\begin{proof}
  Using Theorem \ref{thm:moving-preserves-spine} and the definition of backward moving, one gets
  \[
  \spine(d') \equiv \spine(\move(d',(\overline{d})^{-1})) \equiv \spine(\evom(d',\overline{d})). \qedhere
  \]
\end{proof}

We would like to point out that our running example is chosen purposefully as a typical instance of a situation we found quite often when proving the correctness of reductions between NP-problems within a graph-transformational framework.
This applies – among others – to the well-known reductions of \emph{clique} to \emph{independent set} to \emph{vertex cover}, of \emph{long paths} to \emph{Hamiltonian paths} to \emph{Hamiltonian cycles} to \emph{traveling salesperson}. 
For each of these reductions, one can find graph-transformational models of the two involved NP-problems and the reduction such that the derivations of the source problem can be moved along an initial part of the reduction derivations and the derivations of the target problem can be moved along the remainder part of the reduction derivations.
Then the preservation of the spines is an important part of the correctness proofs.

\section{Conclusion}
\label{sec:conclusion}
In this paper, we have investigated the relationship between two elementary operations on derivations:
moving a derivation along a derivation based on parallel and sequential independence on one hand and restriction of a derivation  with respect to a monomorphism into the start object on the other hand.
As main result, we have shown that moving a derivation preserves its spine being the minimal restriction.
This stimulates further research in at least four directions:

\begin{itemize}
\item
  Sequentialization, parallelization, and shift are three further well-known operations on parallel derivations that are closely related to parallel and sequential independence (cf. \cite{Kreowski:78}).
  Therefore, one may wonder how they behave with respect to the spine.
  Moreover, one may consider the relation between amalgamation and the generalized confluence of concurrent rule applications and spines.
\item
  In~\cite{Kreowski-Kuske-Lye-Windhorst:22}, we have exploited the moving of derivations along derivations and the preservation of spines to prove the correctness of reductions between NP-problems.
  As such reductions are special kinds of model transformations, it may be of interest to explore whether moving and its properties can be used for proving correctness of other model transformations.
  This may be particularly interesting for adhesive model categories where the models are not just graphs.
\item
  We have considered the construction and properties of accessed parts and restrictions of derivations only as far as needed to prove the preservation of spines by the moving of derivations.
  But further investigations may be of interest.
  For example, if a derivation $d'$ is a restriction of a derivation $d$, then $d$ may be considered as extension of $d'$.
  In Kreowski~\cite{Kreowski:78,Kreowski:79}, extension of a graph-transformational derivation is defined independently of restriction, and then it is shown that both operation are inverse to each other.
  So one may wonder whether this can also be done in an adhesive-category context.
  Another question is: How are the accessed parts of a $P$-derivation and a $\overline{P}$-derivation starting in the same object related if the pair of sets of rules $(P,\overline{P})$ is parallel independent?

\item
  In Ehrig et al.~\cite{Ehrig-Ehrig-Prange.ea:06,Ehrig-Ermel-Golas-Hermann:15}, the double-pushout approach is discussed and
investigated in the framework of adhesive HLR categories that is more general
than the framework of adhesive categories. The main difference is that some of
the properties and concepts concern a subclass $\mathcal{M}$ of the class of
monomorphisms with appropriate closure properties. We expect that the
considerations in the present paper work in the more general setting, too.
It may be interesting to look into the technicalities.
\end{itemize}

\noindent
\textbf{Acknowledgement.}
We are grateful to the anonymous reviewers for their detailed comments that led to various improvements.

\bibliographystyle{alphaurl}
\bibliography{main}

\end{document}